\documentclass{amsart}

\usepackage{amssymb,latexsym,amsfonts,amsmath}
\usepackage{graphicx}
\graphicspath{{Figures/}}
\usepackage{booktabs}
\usepackage{xcolor}
\usepackage{epsfig}
\usepackage{subcaption}
\usepackage{dsfont}
\usepackage{epstopdf}

\usepackage{cite}
\usepackage{amsmath,amssymb,amsfonts}
\usepackage{algorithmic}
\usepackage{graphicx}
\usepackage{algorithm,algorithmic}
\usepackage{hyperref}
\usepackage{textcomp}

\usepackage{amsmath}
\usepackage{amsfonts}
\usepackage{amssymb}
\usepackage{color}

\usepackage{url}
\usepackage{breakurl}
\usepackage{hyperref}

\newcommand{\ball}{\Omega}
\def\reals{\mathbb{R}}

\usepackage{graphicx}

\newtheorem{theorem}{Theorem}[section]
\newtheorem{lemma}[theorem]{Lemma}
\newtheorem{proposition}[theorem]{Proposition}

\newtheorem{definition}{Definition}[section]
\newtheorem{example}{Example}[section]
\newtheorem{remark}{Remark}

\newcommand{\alphabeth}{\Sigma_{a}}

\newcommand{\AP}{\mathsf{AP}} 
\newcommand{\conv}{\textrm{conv}}

\newcommand{\until}{\mathbin{\sf U}}

\newcommand{\nex}{\mathord{\bigcirc}}

\newcommand{\word}{\boldsymbol{\omega}}
\newcommand{\wordt}[1]{\boldsymbol{\omega}_{#1}}

\newcommand{\Lab}{\textsf{L}}   

\begin{document}

\begin{abstract}
We consider the notion of resilience for cyber-physical systems, that is, the ability of the system to withstand adverse events while maintaining acceptable functionality. We use finite temporal logic to express the requirements on the acceptable functionality and define the resilience metric as the maximum disturbance under which the system satisfies the temporal requirements. We fix a parameterized template for the set of disturbances and form a robust optimization problem under the system dynamics and the temporal specifications to find the maximum value of the parameter. Additionally, we introduce two novel classes of specifications: closed and convex finite temporal logics specifications, offering a comprehensive analysis of the resilience metric within these specific frameworks. From a computational standpoint, we present an exact solution for linear systems and exact-time reachability and finite-horizon safety, complemented by an approximate solution for finite-horizon reachability. Extending our findings to nonlinear systems, we leverage linear approximations and SMT-based approaches to offer viable computational methodologies. The theoretical results are demonstrated on the temperature regulation of buildings, adaptive cruise control and DC motors.          
\end{abstract}

\title[Temporal Logic Resilience for Dynamical Systems]{Temporal Logic Resilience for Dynamical Systems$^\star$}
\author[A. Saoud]{Adnane Saoud$^{1}$}
\author[P. Jagtap]{Pushpak Jagtap$^{2}$}
\author[S. Soudjani]{Sadegh Soudjani$^{3}$}
\address{$^1$College of Computing, University Mohammed VI Polytechnic, Benguerir, Morocco.}
\email{adnane.saoud@um6p.ma}
\address{$^2$Robert Bosch Center for Cyber-Physical Systems, Indian Institute of Science, Bangalore, India.}
\email{pushpak@iisc.ac.in}
\address{$^3$Max Planck Institute for Software Systems, Kaiserslautern, D-67663, Germany.}
\email{sadegh@mpi-sws.org}
\thanks{$^\star$ This work was supported in part by by the ARTPARK, the Google Research Grant, the SERB Start-up Research Grant, and the CSR Grants by Siemens and Nokia. This project is also supported by the following grants: EPSRC EP/V043676/1, EIC 101070802, and ERC 101089047.}
\maketitle

\section{Introduction}
Resilience has been studied by many research communities and it is broadly defined as \emph{the ability of a system to withstand adverse events while maintaining an acceptable functionality.}
For critical infrastructures, resilience is the main factor determining their reliability and is improved by continuously enhancing the prevention and absorption of disruptive events, and the recovery and adaptation for such events \cite{rehak2019complex}.
For IT systems, resilience is considered mainly against 
adverse cyber events, which are the cyber attacks that negatively impact the availability, integrity, or confidentiality of the system \cite{pillitteri2019developing}.
With climate change increasing the extreme flood events, resilience metrics that consider the dynamical changes of the system have also received attention in the water research community to define and assess the resilience of water resource recovery facilities \cite{holloway2022exploring}.


In this paper, we provide a notion of resilience for cyber-physical systems that integrates the time-evolution of the system with temporal logic to provide a quantitative measure on how the system can cope with disturbances. The temporal logic is used to formally encode the safety and other compliance requirements on the operation of the system and also express the expected behavior of the system to disturbances. We define resilience as the largest disturbance within a given (parameterized) set that can be applied to the system in its time evolution while still satisfying the temporal logic specification. This can also be interpreted as the minimum disturbance that needs to be applied to the system to falsify the specification.

We focus on a discrete-time dynamical model of the system and express resilience requirements as linear temporal logic specifications over finite traces (LTL$_F$) \cite{ijcai2017p189}. We provide different characterizations of the resilience metrics for general LTL$_F$ specifications. Moreover, we introduce two novel classes of LTL$_F$ specifications: closed and convex finite LTL$_F$ specifications. We then characterize the fragments of LTL$_F$ specifications that are closed and convex, and provide an analysis of the structural properties of resilience with respect to closed and convex specifications. From a computation perspective, we show how the optimization for computing resilience can be solved exactly for linear systems with respect to specifications, including the exact-time reachability and finite-horizon safety, and approximately with respect to finite-horizon reachability. Moreover, we provide under-approximations of the resilience metric for nonlinear systems using linear programs and SMT-based approaches. We then provide numerical examples showing the merits of the proposed results.
%

Our definition of resilience is substantially different from \emph{robustness} for temporal specifications, which is defined as follows. Robustness of a system $\Sigma$ with respect to a temporal specification $\phi$ is the largest value $\varepsilon$ such that we still satisfy $\phi$
if we expand the solutions of $\Sigma$ with a uniform $\varepsilon$ over time and over trajectories
\cite{donze2010robust,fainekos2009robustness}.
This definition does not take into account the dynamics of the system $\Sigma$ and directly applies the expansion to the solution of $\Sigma$. In reality, disturbances and extreme events affect the time evolution of the system, and this needs to be integrated with any definition of resilience for dynamical systems.

\smallskip
\noindent\textbf{Related work.}
The literature on defining quantitative semantics for different classes of temporal logic is relatively mature.
These quantitative semantics study how well the system trajectories satisfy a given specification. The techniques include using discounting modalities that give less importance to distant events \cite{almagor2014discounting} and averaging modalities \cite{LTLaverage} where the semantics of standard modalities are extended using min, max, and a long-run average operator. The paper \cite{donze2010robust} considers real-valued signals and presents variants of robustness measures that indicate how far a given signal stands, in space and time, from satisfying or violating a property and studies their sensitivity to the parameters of the system.
The paper \cite{fainekos2009robustness} considers the robust interpretation of Metric Temporal Logic to connect robust satisfaction of properties on discrete-time signals to their continuous-time counterparts.

The authors in \cite{donze2013efficient} present an efficient algorithm for computing the robustness degree in which a piecewise-continuous signal satisfies or violates a Signal Temporal Logic (STL) formula. Application of robustness metrics in specification-based monitoring of cyber-physical systems (CPS) is provided in \cite{bartocci2018specification} with a survey of theory and tools. The robustness metric is also used for temporal logic falsification of CPS \cite{aerts2018temporal}.
STL is also used in \cite{chen2022stl,chen2023stl} to study two important resilience properties of CPS, which are recoverability and durability. Having tolerance thresholds in both time and value is considered in \cite{gazda2020logical} to characterize approximate notions of hybrid conformance. The notion of robustness with respect to a temporal specification is also defined and used in stochastic systems. The paper \cite{ilyes2022stochastic} defines robustness for continuous-time linear stochastic systems with respect to STL specifications. This notion computes the probability of satisfying atomic prepositions in the STL formula and combines the related probabilities using logical connectives based on the structure of the STL formula to get an interval for robustness. The robustness of STL specifications is used in \cite{farahani2018shrinking} to perform controller synthesis for stochastic systems using a model predictive control approach. While we do not consider stochastic disturbances in this work, any available statistics on the likelihood of disturbances can be mapped to statistical properties of the resilience metric, e.g., by constructing a distribution for the resilience value when the distribution of the possible disturbances is known.

All the works mentioned above study the robust satisfaction of properties for a given set of disturbances. In contrast, we are looking at characterizing resilience to compute the largest disturbance within a given parameterized set. The work closest in spirit to our approach is the paper \cite{schulze2017scaling} that is limited to linear systems and safety specifications and studies properties of polytopic invariant sets when the size of the disturbance set changes with a scaling factor. The current paper also extends the preliminary results presented in \cite{resilience}. In the current version, we added numerical examples and a new case study to illustrate the proposed concepts, we provided a new approach to compute the resilience metric for finite-horizon reachability using scenario optimization \cite{esfahani2014performance} and a new approach to compute the resilience metric for nonlinear systems using an SMT-based approach \cite{barrett2018satisfiability}. Finally, we presented a characterization of the fragments of LTL$_F$ specifications that are closed and convex, and we included all the proofs.

\smallskip
In summary, the main contributions of this paper are as follows.
We provide a quantitative resilience metric by integrating the underlying dynamics with temporal logic specifications.
We show how the related optimizations can be solved for linear systems and various types of specifications.
We show that resilience with respect to convex or closed specifications enjoys some \emph{nice} properties.
Finally, we provide under-approximations of the
resilience metric for nonlinear systems. 

The rest of this paper is organized as follows. We provide the preliminaries on the model class and temporal logic specifications in Sec.~\ref{sec:prel}.
We define the notion of resilience and study its structural properties in Sec.~\ref{sec:characterization}.
We provide computations for linear systems in Sec.~\ref{sec:linear} and discuss the computation of lower bounds for nonlinear systems in Sec.~\ref{sec:nonlinear}. The results are demonstrated on examples from adaptive cruise control, temperature regulation of buildings, and DC motors in Sec.~\ref{sec:case_studies} with concluding remarks in Sec.~\ref{sec:concl}.

\section{Preliminaries}
\label{sec:prel}
\textbf{Notation:} 
The symbols $\reals$, $\reals_{\geq 0}$, $\mathbb{N}$, and $\mathbb{N}_{\geq n}$ denote the set of real, nonnegative real numbers, nonnegative
integers, and integers greater than or equal to $n\in\mathbb{N}$, respectively.
We use $\reals^{n\times m}$ to denote the space of real matrices
with $n$ rows and $m$ columns.
For a matrix $A\in \reals^{n\times m}$, $A^T$ represent the transpose of $A$. For a vector $x\in\reals^n$, we use $\|x\|$ and $\|x\|_\infty$ to denote the Euclidean and infinity norm, respectively. 
We use $\mathbb{I}$ to denote the identity matrix. 
For a set of $p$ points $C = \{c_1,c_2,\ldots,c_p\}$, $c_i\in \reals^n$, the convex hull of $C$ is represented by $\conv(c_1,c_2,\ldots,c_p):=\{\alpha_1c_1+\alpha_2c_2+\ldots+\alpha_pc_p\mid c_i\in C, \alpha_i\geq 0, i\in\{1,2,\ldots,p\},\sum_{i=1}^{p}\alpha_i=1\}$. An interval in $\mathbb{R}^n$ is a set denoted by $X=[\underline{x}_1,\overline{x}_1]\times[\underline{x}_2,\overline{x}_2]\times \ldots \times [\underline{x}_n,\overline{x}_n]$ and defined as $X=\left\{x \in \mathbb{R}^n \mid \underline{x}_i \leq x_i \leq \overline{x}_i,~i\in \{1,2,\ldots,n\}\right\}$, where $x_i \in \mathbb{R}$ represents the ith component of the vector $x \in \mathbb{R}^n$. In particular, when $\underline{x}_i=\underline{x}$ and $\overline{x}_i=\overline{x}$ for all $i \in \{1,2,\ldots,n\}$, then the interval $X$ can be written in a compact form as: $X=[\underline{x},\overline{x}]^n$. Given $x \in \mathbb{R}^n$ and $\varepsilon \geq 0$, $\ball_{\varepsilon}(x)=\{z \in \mathbb{R}^n \mid \|z-x\|_\infty \leq \varepsilon$\}.

\subsection{Discrete-Time Dynamical Systems}
A discrete-time system is a tuple $\Sigma = (X,D,f)$, where
$X\subset \reals^n$ is the state space and
$D\subset\reals^n$ is the disturbance space which is assumed to be a compact set containing the origin. The evolution of the state of $\Sigma$ is given by
\begin{equation}
\label{eqn:sys}
x(k+1)=f(x(k))+d(k),~~ k\in \mathbb{N},
\end{equation}
where $d(k)\in D$ represents the additive disturbance. The trajectory of system $\Sigma$ of length $N$ is represented by $w_x=(x_0,x_1,\ldots,x_{N-1}) \in X^{N}$, where $x_k$ represents the value of trajectory starting from a state $x(0)=x_0\in X$ at $k^{\text{th}}$ instance (i.e., $x(k)$).

\subsection{LTL$_F$ Specifications}

Linear temporal logic (LTL) provides a high-level language for describing the desired behavior of a dynamical system. Formulas in this logic are constructed inductively using a set of atomic propositions and combining them via Boolean operators \cite{baier2008principles}.
Consider a finite set of atomic propositions $\AP$ that defines the alphabet $\alphabeth := 2^{\AP}$. Each letter of this alphabet evaluates a subset of the atomic propositions as true. In this work, we consider LTL specifications over finite words, referred to as LTL$_F$, where the letters form finite words defined as
$\word=(\wordt{0},\wordt{1},\wordt{2},\ldots,\wordt{N-1})\in\alphabeth^{N}$ for some $N\in\mathbb{N}$, with $\wordt{i}$ representing the letter of the word at $i$th instance.
These words are connected to trajectories of the system $\Sigma$ via a measurable labeling function $\Lab:X\rightarrow \alphabeth$ that assigns letters $\alpha =\Lab(x)$ to state $x\in X$. That is, any finite trajectory $w_x = (x_0,x_1,\ldots,x_{N-1})$ is mapped to the set of finite words $\alphabeth^{N}$, as
$\word=\Lab(w_x) := (\Lab(x_0),\Lab(x_1),\Lab(x_2),\ldots\Lab(x_{N-1}))$.
\medskip
\begin{definition}
	\label{def:LTL}
	An LTL$_F$ formula over a set of atomic propositions $\AP$ is constructed inductively as
	\begin{equation*}
	\label{eq:PNF}
	\psi ::=  \textsf{true} \,|\, p \,|\, \neg \psi  \,|\,\psi_1 \wedge \psi_2 \,|\, \psi_1 \vee \psi_2 \,|\, \nex \psi \,|\, \psi_1\until \psi_2 \,|\, \square\psi\,|\,\lozenge\psi,
	\end{equation*}
	with $p\in \AP$, and $\psi_1,\psi_2,\psi$ being LTL$_F$ formulas.
	\end{definition}

Given a finite word $\word$ of length $N$ and an LTL$_F$ formula $\psi$, we inductively define when an LTL$_F$ formula is true at the $i^{\text{th}}$ step $(i < N)$ and denoted by $\word_{i}\models\psi$, as follows:

\begin{itemize}
\item $\word_i\vDash\textsf{true}$ always hold and $\word_i\vDash\textsf{false}$ does not hold.
\item An atomic proposition, $ \word_i\vDash   p$  for $ p\in \AP$ holds if $p \in\word_{i}$.
\item A negation, $\word_i\vDash\neg p$, holds if $ \word_i\nvDash p$.
\item A logical conjunction, $\word_i\vDash \psi_1\wedge\psi_2$, holds
if $ \word_i\vDash \psi_1$ and $ \word_i\vDash \psi_2$.
\item A logical disjunction, $\word_i\vDash \psi_1\vee\psi_2$, holds
if $ \word_i\vDash \psi_1$ or $ \word_i\vDash \psi_2$.
\item A temporal next operator, $\word_i\vDash\nex\psi$, holds if $\word_{i+1}\vDash \psi$. Similarly, for $0 \leq j < N-i$, $\word_i\vDash\nex^j\psi$, holds if $\word_{i+j}\vDash \psi$. 
\item A temporal until operator, $\word_i\vDash \psi_1\until\psi_2$, holds if for some $m$ such that $i\leq m< N$, we have $\word_m\vDash\psi_2$ and for all $i\leq k< m$, we have $\word_k\vDash\psi_1$.
\item A temporal always operator, $\word_i\vDash \square\psi$, holds if for all $m$ such that $i\leq m< N$, we have $\word_m\vDash \psi$. Similarly, for $0 \leq j < N-i$, $\word_i\vDash \square^j\psi$ holds if for all $i \leq m \leq i+j$, we have $\word_m\vDash \psi$.
\item A temporal eventually operator, $\word_i\vDash \lozenge\psi$, holds if for some $m$ such that $i\leq m< N$, we have $\word_m\vDash \psi$. Similarly, for $0 \leq j < N-i$, $\word_i\vDash \lozenge^j\psi$, holds if for some $m$ such that $i\leq m< i+j$, we have $\word_{m}\vDash \psi$. 
\end{itemize}


An LTL$_F$ formula $\psi$ is true for $\word$, denoted by $\word\vDash\psi$, if and only if $\wordt{0}\vDash\psi$. For a trajectory $w_x=(x_0,x_1,\ldots,x_{N-1}) \in X^{N}$, we say that $w_x \vDash \psi $ if for $\word=\Lab(w_x) := (\Lab(x_0),\Lab(x_1),\ldots,\Lab(x_{N-1}))$, we have $\word \vDash \psi$. Similarly, for a set of trajectories $\mathcal{X} \subseteq X^{N}$, we say that $\mathcal{X} \vDash \psi$, if $w_x \vDash \psi$ for all $w_x \in \mathcal{X}$.


\medskip
\begin{remark}
Our notion of resilience is general, but for computational purposes, we restrict ourselves to the following specifications over words of length $N$: $\square p$, $\nex p$, $\lozenge p$, with $p\in \AP$, and conjunctions over them. Note that all LTL$_F$ can be represented using Deterministic Finite Automata (DFA) \cite{zhu2019first} and effectively, one can represent them using sequences of reach and avoid specifications (i.e., $\psi=\lozenge p \wedge \square \neg q$, where $p,q\in\alphabeth$) \cite[Section III.b]{wang2022verified}. Thus, one can easily use the results provided in the paper for any LTL$_F$ specification using the properties provided in Proposition \ref{prop:struct}.
\end{remark}

\section{Characterizations of Resilience}
\label{sec:characterization}

The goal of the paper is to provide characterizations and algorithmic procedures for computing resilience with respect to different classes of systems and specifications.

\subsection{Resilience for LTL$_F$ Specifications}
Consider the system $\Sigma$ in \eqref{eqn:sys}, with a set of disturbances given by a ball with respect to infinity norm centered at zero: $D := \ball_{\varepsilon}(0)$.
We denote by $\xi(x,\varepsilon)$ the set of trajectories starting from some $x\in X$ with such a disturbance set:
\begin{equation}
\label{eqn:reach}
   \hspace{-0.2em} \xi(x,\varepsilon)\hspace{-0.2em} := \hspace{-0.2em}\left\{(x_0,x_1,x_2,\ldots)\,|\, x_0 \hspace{-0.2em}= \hspace{-0.2em}x, \, x_{k+1} \hspace{-0.2em}\in\hspace{-0.2em} f(x_k)\hspace{-0.2em}+\hspace{-0.2em}D \right\}.
\end{equation}
Note that $\xi(x,0)$ contains only the disturbance-free trajectory of the system (nominal trajectory). Similarly, given a set $A \subseteq X$, we define the set of trajectories $\xi(A,\varepsilon)=\bigcup\limits_{x \in A} \xi(x,\varepsilon).$

\begin{definition}
\label{def:resilience}
Consider the dynamical system $\Sigma$ in \eqref{eqn:sys}, an LTL$_F$ specification $\psi$ and a point $x \in X$. We define the {\it resilience} of the system $\Sigma$ with respect to the initial condition $x$ and the specification $\psi$ as a function $g_{\psi}:X\rightarrow \reals_{\ge 0}\cup\{+\infty\}$ with:
\begin{equation}
\label{eq:quan}\hspace{-0.2em} g_\psi(x) = \begin{cases}
\sup\left\{\varepsilon\ge 0\,|\,\xi(x,\varepsilon)\vDash \psi\right\}, & \text{ if } \xi(x,0)\vDash\psi\\
0, & \text{ if } \xi(x,0)\nvDash\psi,
\end{cases}
\end{equation}
where $\xi(x,\varepsilon)\vDash \psi$ indicates that all trajectories in $\xi(x,\varepsilon)$ satisfy the specification.
\end{definition}

Similarly, for a set $A \subseteq X$, the resilience metric over the set $A$ and specification $\psi$ can be defined as follows:
\begin{equation}
\hspace{-0.2em} g_\psi(x) = \begin{cases}
\sup\left\{\varepsilon\ge 0\,|\,\xi(x,A)\vDash \psi\right\}, & \text{ if } \xi(x,A)\vDash\psi\\
0, & \text{ if } \xi(x,A)\nvDash\psi,
\end{cases}
\end{equation}

\begin{remark}
The above definition assigns a non-negative value to resilience; hence we call it a \emph{resilience metric}. We can extend this definition to take both positive and negative values depending on the satisfaction of the specification by the nominal trajectory $\xi(x,0)$.
For a given $\psi$, let us define the function $h_{\psi}:X\rightarrow \reals\cup\{\pm\infty\}$ with
\begin{equation}
\label{eq:quan2}
h_\psi(x) := \begin{cases}
\sup\left\{\varepsilon\ge 0\,|\,\xi(x,\varepsilon)\vDash \psi\right\}, & \text{ if } \xi(x,0)\vDash\psi\\
 -\sup\left\{\varepsilon\ge 0\,|\,\xi(x,\varepsilon)\vDash \neg\psi\right\}, & \text{ if } \xi(x,0)\not\vDash\psi.
\end{cases}
\end{equation}
Note that
\begin{equation}
\label{eq:quan3 }
h_\psi(x) = \begin{cases}
g_\psi(x), & \text{ if } \xi(x,0)\vDash\psi\\
-g_{\neg\psi}(x), & \text{ if } \xi(x,0)\not\vDash\psi.
\end{cases}
\end{equation}
This definition is useful for guiding the system towards satisfying the specification and then maximizing the resilience by appropriately quantifying disturbances and decision variables. In this paper, we study the resilience of autonomous systems (i.e., systems without control inputs or decision variables). Therefore, we use Definition~\ref{def:resilience} in the rest of the paper.  
\end{remark}

\subsection{Structural Properties of Resilience}
In this subsection, we prove the structural properties of the resilience metric in Definition.~\ref{def:resilience} utilizing the inductive definition of temporal specifications.


We first have the following set of properties of the resilience.

\begin{proposition}
\label{prop:struct}
Consider the dynamical system $\Sigma$ in \eqref{eqn:sys}, an LTL$_F$ specification $\psi$ and a point $x \in {X}$. The following properties hold:
\begin{itemize}
    \item[(i)] When $\psi$ is the $\textsf{true}$ specification, $g_\psi(x) = +\infty$, for all $x\in X$.
    
    \item[(ii)] When $\psi$ is the $\textsf{false}$ specification, $g_\psi(x) = 0$, for all $x \in X$.

    \item[(iii)] For any specification $\psi = \psi_1\wedge\psi_2$, we have that $g_\psi(x) = \min\{g_{\psi_1}(x),g_{\psi_2}(x)\}$, for all $x \in X$.
    
    \item[(iv)] For any specification $\psi = \psi_1 \vee \psi_2$, we have that $g_\psi(x) \geq \max\{g_{\psi_1}(x),g_{\psi_2}(x)\}$, for all $x \in X$. 
    
    \item[(v)] For any specification $\psi = \neg\phi$, we have that $g_\psi(x)g_\phi(x) = 0$, for all $x \in X$.

    \item[(vi)] For $A \subset X$, $g_{\psi}(X)= \inf_{x \in X}g_{\psi}(x)$.

\end{itemize}
\end{proposition}

\begin{proof}
Properties (i), (ii), and (vi) are the direct consequence of Definition \ref{def:resilience}.

\textit{Proof for property (iii):}
For $x\in X$, consider sets $A=\{\epsilon\geq 0\mid \xi(x,\epsilon)\vDash\psi_1\}$ and $B=\{\epsilon\geq 0\mid \xi(x,\epsilon)\vDash\psi_2\}$. Then we can write $A\cap B=\{\epsilon\geq 0\mid \xi(x,\epsilon)\vDash\psi_1 \text{ and }\xi(x,\epsilon)\vDash\psi_2\}$. By the fact that sets $A$ and $B$ are bounded from below, we have $\sup\{A\cap B\}= \min\{\sup A, \sup B\}$. This implies $g_{\psi_1\wedge\psi_2}(x) = \min\{g_{\psi_1}(x),g_{\psi_2}(x)\}$.

\textit{Proof for property (iv):} For $x\in X$, consider sets $A=\{\epsilon\geq 0\mid \xi(x,\epsilon)\vDash\psi_1\}$ and $B=\{\epsilon\geq 0\mid \xi(x,\epsilon)\vDash\psi_2\}$. Then we can write $A\cup B=\{\epsilon\geq 0\mid \xi(x,\epsilon)\vDash\psi_1 \text{ or }\xi(x,\epsilon)\vDash\psi_1\}$. By the fact that $\sup\{A\cup B\}\geq \max\{\sup A, \sup B\}$, one has $g_{\psi_1\vee\psi_2}(x)\geq \max\{g_{\psi_1}(x),g_{\psi_2}(x)\}$.

\textit{Proof for property (v):} 
We consider two cases:\\ 
\textit{case (i):} for $x\in X$, consider $\xi(x,0)\vDash\phi$, then $g_{\psi=\neg\phi}(x)=0$. Thus $g_\psi(x)g_\phi(x) = 0$.\\
\textit{case (ii):} for $x\in X$, consider $\xi(x,0)\vDash\psi=\neg\phi$, then $g_{\phi}(x)=0$. Thus $g_\psi(x)g_\phi(x) = 0$. \\
This concludes the proof.
\end{proof}

We now provide an example showing that equality in the property (iv) does not hold in general.

\begin{example}
Consider the linear system $\Sigma$ with
\begin{align*}
    x_1(k+1)&=x_1(k)+3x_2(k)+d_1(k)\\
    x_2(k+1)&=3x_1(k)-2x_2(k)+d_2(k)
\end{align*}
Consider the sets $S_1=[2,3]\times [1,2]$, $S_2=[2,3]\times [2,3]$ and $S=S_1\cup S_2$. Consider the specifications $\psi_1=\nex S_1$, $\psi_2= \nex S_2$ and $\psi=\psi_1 \vee \psi_2 = \nex (S_1 \cup S_2)$. For the state $x(0)=(1,0.5)$, since $x(1)=(2.5,2)$ is both on the boundaries of the sets $S_1$ and $S_2$, one can easily check that $g_{\psi}(x)=0.5 > \max(g_{\psi_1}(x),g_{\psi_2}(x))=0$.
\end{example}

\medskip


\subsection{Resilience Properties for Closed Specifications}

Now, we provide sufficient conditions on the dynamics of the system $\Sigma$ and the specification $\psi$ allowing us to replace the \textit{$\sup$} operator with the \textit{$\max$} operator in the definition of $g_{\psi}$ in (\ref{eq:quan}), which makes the computation of the resilience metric tractable. To do this, we introduce the class of closed specifications defined below.
Given a sequence of trajectories $w_{x,i}$, $i \in \mathbb{N}$, with  $w_{x,i}=x_{0,i},x_{1,i},x_{2,i}\ldots$, the limit of $w_{x,i}$ is defined by $w_x=\lim_{i\rightarrow \infty} w_{x,i}=x_{0},x_{1},x_{2}\ldots$, where for all $j\in \mathbb{N}$, $x_j=\lim_{i\rightarrow \infty} x_{j,i}$. The sequence $w_{x,i}$ is called converging whenever $w_x$ exists. Hence, the limit of a sequence of trajectories can be seen as an element-wise limit of its components.

\begin{definition}[Closed specification]
\label{def:closed}
Consider a set $X \subseteq \mathbb{R}^n$, an LTL$_F$ formula $\psi$ is said to be closed if the following holds: for any converging sequence of trajectories $w_{x,i}$, $i \in \mathbb{N}$, if for all $i \in \mathbb{N}$, $w_{x,i} \vDash \psi$, then $w_x=\lim_{i\rightarrow \infty} w_{x,i} \vDash \psi$.
\end{definition}

Intuitively, the closedness property states that a specification is preserved when going from a sequence of trajectories to its element-wise limit. In Appendix \ref{sec:closed}, we provide different characterizations of the fragment of LTL$_F$ specifications that are closed.

Before stating the main result, we have the following auxiliary lemma.

\begin{lemma}
\label{lem:disturb1}
Consider the discrete-time system $\Sigma$ in \eqref{eqn:sys}, an $\varepsilon \geq 0$ and an LTL$_F$ formula $\psi$. If $\xi(x,\varepsilon) \vDash \psi$ then  $\xi(x,\varepsilon') \vDash \psi$ for any $\varepsilon' \leq \varepsilon$.
\end{lemma}

\begin{proof}
The proof follows directly from the fact that $\xi(x,\varepsilon') \subseteq \xi(x,\varepsilon)$ for $\varepsilon' \leq \varepsilon$.
\end{proof}

\begin{proposition}
Consider the discrete-time system $\Sigma$ in \eqref{eqn:sys} defined on a metric space $X$. Consider $x \in X$ and an LTL$_F$ specification $\psi$. If the map $f:\mathbb{R}^n \rightarrow \mathbb{R}^n$ is continuous and if $\psi$ is closed, then:
\begin{equation*}
g_\psi(x) = \begin{cases}
\max\left\{\varepsilon\ge 0\,|\,\xi(x,\varepsilon)\vDash \psi\right\}, & \text{ if } \xi(x,0)\vDash\psi,\\
0, & \text{ if } \xi(x,0)\nvDash\psi.
\end{cases}
\end{equation*}

\end{proposition}
\medskip
\begin{proof}
First using Theorem 2.28 of~\cite{rudin1976principles}, the supremum is equal to the maximum if the set $A=\left\{\varepsilon\ge 0\,|\,\xi(x,\varepsilon)\vDash \psi\right\}$ is closed. To show the result, it is sufficient to show that the set $A$ is closed. Consider a sequence $\varepsilon_i \in \mathbb{R}_{\geq 0}$, $i \in \mathbb{N}$, assume that for all $i \in \mathbb{N}$, $\xi(x,\varepsilon_i) \vDash \psi$ and let us show that $\xi(x,\varepsilon) \vDash \psi$, with $\varepsilon = \lim_{i \rightarrow \infty } \varepsilon_i$. Using the closedness of the specification $\psi$, one has that $ \lim_{i \rightarrow \infty }\xi(x,\varepsilon_i) \vDash \psi$, which implies from the continuity of the map $f$ and using Lemma~\ref{lem:continuity} (in Appendix) that $\lim_{i \rightarrow \infty }\xi(x,\varepsilon_i) \subseteq \xi(x,\lim_{i \rightarrow \infty } \varepsilon_i) = \xi(x,\varepsilon) \vDash  \psi$.
\end{proof}

We also provide sufficient conditions under which the map $g_{\psi}(x)$ remains bounded for an $x \in X$.

\begin{proposition}
Consider the dynamical system $\Sigma$ in \eqref{eqn:sys}, an LTL$_F$ specification $\psi$ and a point $x \in X$. If the set $\Lab^{-1}(\psi) \subseteq \mathbb{R}^n$ is a bounded subset of $\mathbb{R}^n$, then $g_{\psi}(x)$ is bounded.
\end{proposition}
\begin{proof}
Assume that $g_{\psi}(x)$ is unbounded, hence the set $\xi(x,g_{\psi}(x))$ $\subseteq \Lab^{-1}(\psi)\subseteq \mathbb{R}^n $ is unbounded, which contradicts the fact that the set $\Lab^{-1}(\psi) \subseteq \mathbb{R}^n$ is bounded. Hence, $g_{\psi}(x)$ is bounded.
\end{proof}

\subsection{Resilience Properties for Linear Systems and Convex Specifications}

In this part, we introduce the concept of convex specifications and present an efficient approach for the computation of the resilience metric for the class of linear systems and convex specifications.

\begin{definition}[Convex specification]
\label{def:convex} An LTL$_F$ formula $\psi$ is said to be convex if the following holds: for trajectories $w_{x,i}$, $i \in \{1,2\}$, if $w_{x,i} \vDash \psi $, then for any trajectory $w_x=\lambda w_{x,1} + (1-\lambda) w_{x,2}$, $\lambda \in [0,1]$, we have that $w_x \vDash \psi$.
\end{definition}
Intuitively, the convexity property states that the specification is preserved under a convex hull operator. In Appendix \ref{sec:convex} provides a characterization of the fragment of LTL$_F$ specifications that are convex. 

Now, consider the case where the objective is to compute $g_{\psi}(A)$ for a set $A \subset \mathbb{R}^n$. A straightforward approach that was mentioned earlier in the property (vi) of Proposition \ref{prop:struct} is to use the fact that $g_{\psi}(A) = \inf_{x \in A}g_{\psi}(x)$, which requires computing $g_{\psi}(x)$ for all $x \in X$ and that can be computationally infeasible for continuous sets. In this part, we present an efficient approach to compute $g_{\psi}(A)$ for the case where the set $A$ can be written as a convex closure of a finite number of points. 
Our result relies on the superposition principle for linear systems and the introduced class of convex specifications.

\begin{theorem}
\label{thm:conv}
Consider the discrete-time linear system in \eqref{eqn:sys}. Consider $A=\conv(c_1,c_2,\ldots,c_p) \subset \mathbb{R}^n$ and a convex specification $\psi$. We have $g_{\psi}(A)=\min\limits_{i=1,2,\ldots,p} g_{\psi}(c_i)$.
\end{theorem}
\begin{proof}
First, one has from (vi) in Proposition \ref{prop:struct} that $g_{\psi}(X) \leq \min\limits_{i=1,2,\ldots,p} g_{\psi}(c_i)$. To show that $g_{\psi}(A) \geq \min\limits_{i=1,2,\ldots,p} g_{\psi}(c_i)$ , let us denote $\varepsilon=\min\limits_{i=1,2,\ldots,p} g_{\psi}(c_i)$ and consider $x \in A=\conv(c_1,c_2,\ldots,c_p)$. We have the existence of $\alpha_1,\alpha_2,\ldots,\alpha_p \geq 0$ such that $\sum\limits_{i=1}^p \alpha_ic_i=x$ and $\sum\limits_{i=1}^p \alpha_i=1$.
To show the result, let us show that 
\begin{equation}
\label{eqn:thm}
\xi(x,\varepsilon) \subseteq \conv\left(\xi(c_1,\varepsilon),\xi(c_2,\varepsilon),\ldots,\xi(c_p,\varepsilon)\right).
\end{equation}

Consider a state trajectory $w_x=x_0,x_1,x_2,\ldots \in \xi(x,\varepsilon)$, we have the existence of $d_0,d_1,d_2,\ldots \in \ball_{\varepsilon}(0)$ such that $x_0=x$ and $x_{k+1}=Ax_k+d_k$ for all $k \in \{0,\ldots,N-1\}$. Now for $i \in \{1,2,\ldots,p\}$, consider $w_{x,i}=x_{0,i},x_{1,i},x_{2,i},\ldots \in \xi(c_i,\varepsilon)$ defined as follows: $x_{0,i}=c_i$ and $x_{k+1,i}=Ax_{k,i}+d_k$. Hence, using the fact that $\sum\limits_{i=1}^p \alpha_i=1$ one has that for all $i \in \{1,2,\ldots,p\}$ and for all $k \in \{0,\ldots,N-1\}$, $$x_{k}=\sum\limits_{i=1}^p (\alpha_ix_{k,i}+d_k) \in \conv({x_{k,1},x_{k,2},\ldots,x_{k,p}}),$$ which in turn implies that $w_x \in \conv({w_{x,1},w_{x,2},\ldots,w_{x,p}})$ and \eqref{eqn:thm} holds. 

Moreover, since $\varepsilon \leq g_{\psi}(c_i)$, $i\in \{1,2,\ldots,p\}$, we have from Lemma~\ref{lem:disturb1} that for all $i \in \{1,2,\ldots,p\}$, $\xi(c_i,\varepsilon)\vDash \psi$. Hence using the fact that $\xi(x,\varepsilon) \subseteq \conv({\xi(c_1,\varepsilon),\xi(c_2,\varepsilon),\ldots,\xi(c_p,\varepsilon)})$ and from the convexity of the specification $\psi$, we conclude that $\xi(x,\varepsilon)\vDash \psi$. Hence, $g_{\psi}(x)=\min\limits_{i=1,2,\ldots,p} g_{\psi}(c_i)$.
\end{proof}

In Sections~\ref{sec:linear}-\ref{sec:nonlinear}, we provide results on the computation of the resilience metric for linear and nonlinear systems for various specifications.

\section{Computation of Resilience for Linear Systems}
\label{sec:linear}

In this section, we show how the resilience metric can be computed exactly for some classes of specifications, such as the exact time reachability and finite-horizon safety. Moreover, we show how it can be approximated, up to any accuracy, for finite-horizon reachability properties.

\subsection{Exact-Time Reachability}
\label{sec:linear_reach}
Consider the linear discrete-time system $\Sigma$ defined by:
\begin{equation}
\label{eq:linear}
    x(k+1) = A x(k) + d(k),
\end{equation}
with $x(k),d(k) \in \mathbb{R}^n, k \in\mathbb{N}$ and reachability at a specific time $N\in\mathbb{N}$ as $\psi = \nex^N \Gamma$, for some polytopic set $\Gamma$.
We have the following result showing that the computation of the resilience metric for linear systems and reachability at a specific time point specification boils down to a linear optimization problem. 
\begin{theorem}
\label{thm:lin_reachN}
Consider the linear system $\Sigma$ in \eqref{eq:linear} and the specification $\psi=\nex^N \Gamma$, for $N \in \mathbb{N}$, where $\Gamma$ is a polytope $\Gamma = \{x\in X\,|\, Gx\le H\}$ with $G \in \mathbb{R}^{q \times n} $ and $H \in  \mathbb{R}^q$, $q \in \mathbb{N}$. We have
\begin{equation}
\label{eqn:opt1}
g_{\psi}(x) \!=\! \min \{\varepsilon\ge 0\,|\, P\ge 0,P A_b \!=\! E, P B_b \le \varepsilon F(x)
\}
\end{equation}
with
\begin{itemize}
    \item $
    A_b := \left[\begin{array}{cc}
         \mathbb I \\
         - \mathbb I
    \end{array}
    \right] \in \mathbb{R}^{2nN\times n N} $ and 
   $B_b := \left[\begin{array}{c}
         \mathbf 1 \\
         \mathbf 1
    \end{array}
    \right] \in \mathbb{R}^{2nN}$
\item $E=[GA^{N-1}, GA^{N-2}, \ldots, GA, G] \in \mathbb{R}^{q\times nN}$
\item $F(x)=H - GA^Nx \in \mathbb{R}^q$.
\end{itemize}
\end{theorem}
\begin{proof}
We have that 
\begin{align*}
\begin{cases}
g_\psi(x) = 0\qquad \text{ if } \quad GA^Nx > H, \text{ else: }\\
g_\psi(x) = \max\,\{\varepsilon\ge 0,\,s.t.\, \text{ for all }d_0,\ldots,d_{N-1}\in\ball_\varepsilon(0)\\
\qquad\qquad\qquad G(A^Nx+A^{N-1}d_0+\ldots+d_{N-1})\le H.\}
\end{cases}
\end{align*}

Hence, the computation of the resilience metric $g_{\psi}$ requires solving the following optimization problem
\begin{align*}
\max \{\varepsilon\ge 0\,|\, E Y\le \frac{1}{\varepsilon}F(x),\text{for all } Y \text{ satisfying } A_b Y\le B_b\},
\end{align*}
with $Y:=\frac{1}{\varepsilon}[d_0;d_1;\ldots,d_{N-1}]$ being the free variables,
$F(x) := H - GA^Nx$,
$E:=[GA^{N-1}, GA^{N-2}, \ldots, GA, G]$,
and $A_b Y\le B_b$ represents the inequality $\|Y\|_\infty\le 1$ in matrix form using the matrices $A_b$ and $B_b$ defined above.\footnote{In the general case when the disturbance set has a template in the form of a polytope $\{w\,|\, A_w w\le B_w\}$, we have
$A_b = \text{diag}(A_w,A_w,\ldots,A_w)$ and
    $B_b =[B_w;B_w;\ldots;B_w]$.}
\begin{equation*}
    A_b = \left[\begin{array}{cc}
         \mathbb I  \\
         - \mathbb I
    \end{array}
    \right],
    B_b = \left[\begin{array}{c}
         \mathbf 1 \\
         \mathbf 1
    \end{array}
    \right].
\end{equation*}

According to the affine form of Farkas' Lemma \cite{schrijver1998theory}, see Theorem~\ref{thm:farkas} in Appendix,
one gets the following optimization problem:
\begin{equation*}
g_{\psi}(x) = \max \{\varepsilon\ge 0\,|\, P\ge 0, P A_b = E, P B_b \le \frac{1}{\varepsilon} F(x)
\},
\end{equation*}
which is equivalent to the optimization problem in \eqref{eqn:opt1}.
\end{proof}

\medskip

We have the following illustrative example for the results of Theorems~\ref{thm:lin_reachN} and~\ref{thm:conv}.

\begin{example}
\label{examp:1}
Consider the linear dynamical system $\Sigma$ in \eqref{eq:linear} with 
$A=\begin{pmatrix}
0.1 & -1 \\
-0.5 & -0.2
\end{pmatrix}$
and the polytope $\Gamma = \{x\in X\,|\, Gx\le H\}$ with
$$G=\begin{pmatrix}
0 & -0.2747 \\
0.3714 & 0 \\
0 & 0.3714 \\
-0.2747 & 0
\end{pmatrix},\text{ and }
H=\begin{pmatrix}
0.9615 \\
0.9285 \\
0.9285 \\
0.9615
\end{pmatrix}.$$
Consider the specification defined by $\psi = \nex^3 \Gamma$. In simple words, the objective is to reach the set $\Gamma$ in three time steps, the set $\Gamma$ is defined as an interval $\Gamma=[-3.5,3.5]\times[-2.5,2.5]$. Figure~\ref{fig:robust} represents the resilience metric $g_\psi(x)$ on the domain $A=[-4,6]\times[-4,6]$. Moreover, the numerical results show that $g_{\psi}(X)=0.1166$. On the other hand, we have that the set $A=[-4,6]\times[-4,6] = \conv(c_1,c_2,c_3,c_4)$, with $c_1=(-4,-4)$, $c_2=(-4,6)$, $c_3=(6,-4)$ and $c_4=(6,6)$. The resilience metrics of the vectors $c_1$, $c_2$, $c_3$ and $c_4$ are $g_{\psi}(c_1)=0.417$, $g_{\psi}(c_2)=0.1166$, $g_{\psi}(c_3)=0.1380$ and $g_{\psi}(c_4)=0.1203$. One can easily check that $g_{\psi}(X)=\min\limits_{i=1,2,3,4} g_{\psi}(c_i)=g_{\psi}(c_2)$, which is consistent with the result of Theorem~\ref{thm:conv}.
\end{example}
\begin{figure}[!t]
	\begin{center}
		\includegraphics[scale=0.45]{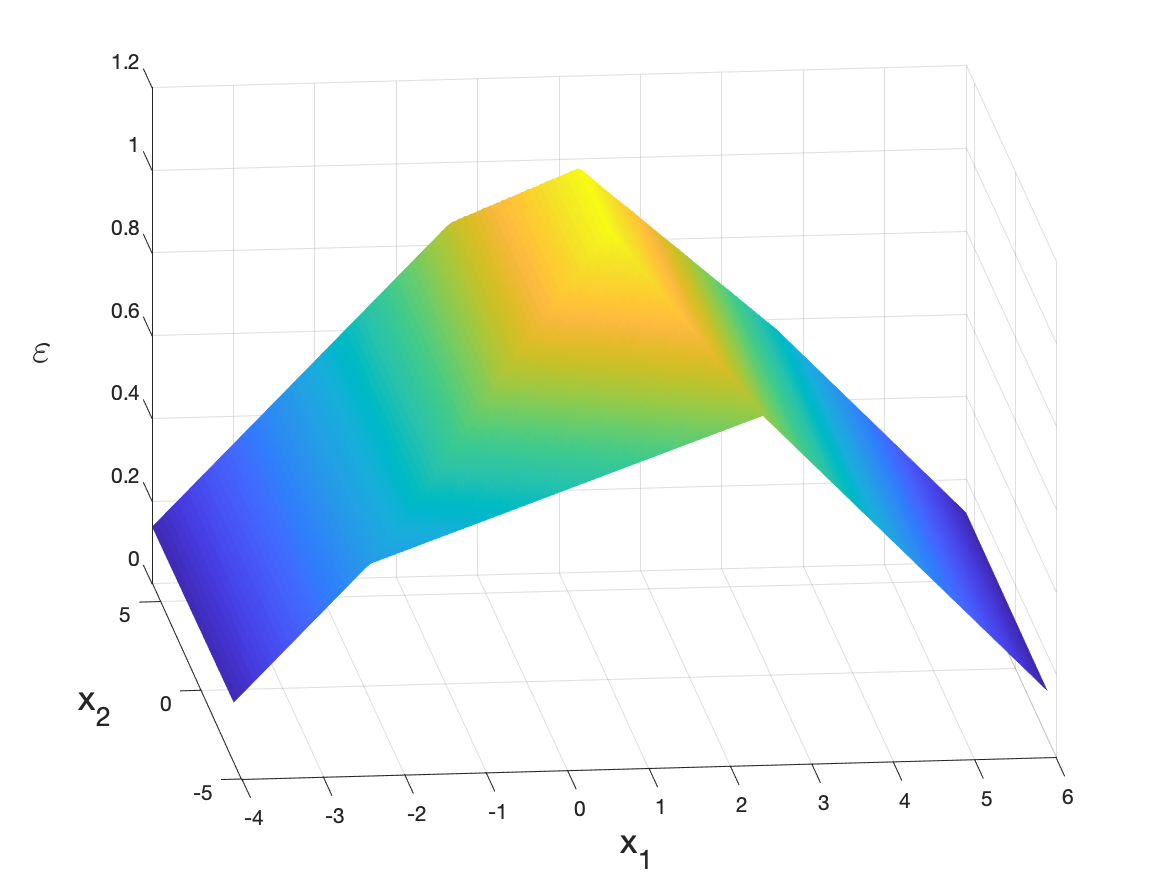}\\
	\end{center}
        \vspace{-0.4cm}
	\caption{Function $g_{\psi}(x)$ for Example \ref{examp:1}.}
	\label{fig:robust}	
\end{figure}
\subsection{Finite-Horizon Safety Specifications}
In this section, we provide a closed-form expression of the resilience metric for the case of linear systems and finite-horizon safety specifications.  
\begin{theorem}
\label{thm:fin_safe}
Consider the linear system $\Sigma$ in \eqref{eq:linear} with finite-horizon safety specification $\psi = \square^N \Gamma$, where $\Gamma=\{x\in X\,|\, Gx\le H\}$, for some $N \in \mathbb{N}$. We have
\begin{align*}
& g_{\psi}(x) = \max \{\varepsilon \geq 0 \mid P\ge 0, P A_b = E, P B_b \le \varepsilon F(x)\},\\
\end{align*}
\begin{itemize}
    \item
    $
    A_b = \left[\begin{array}{cc}
         \mathbb I  \\
         - \mathbb I
    \end{array}
    \right] \in \mathbb{R}^{2nN\times n N} $ and 
   $B_b = \left[\begin{array}{c}
         \mathbf 1 \\
         \mathbf 1
    \end{array}
    \right] \in \mathbb{R}^{2nN}$   
\item
\begin{align*}
& E = \left[
\begin{array}{cccccc}
G & 0 & 0 & \ldots & 0 & 0\\
GA & G & 0 & \ldots & 0 & 0\\
GA^2 & GA & G & \ldots & 0 & 0\\
\vdots & \vdots & \vdots & \vdots & \vdots & \vdots\\
GA^{N-1} & GA^{N-2} & GA^{N-3} & \ldots & GA & G  
\end{array}
\right],\\
& F(x) =[H - GAx \; H - GA^2x \; \ldots \; H - GA^Nx].
\end{align*}
\end{itemize}
\end{theorem}
\begin{proof}
The resilience metric is given by
\begin{align*}
g_\psi(x) = &\max\,\varepsilon\ge 0,
\text{s.t for all }d_0,\ldots,d_{N-1}\in\ball_\varepsilon(0),\\
&\begin{cases}
Gx\le H\\
G(Ax+d_0)\le H\\
G(A^2x+Ad_0+d_1)\le H\\
\vdots\\
G(A^Nx+A^{N-1}d_0+\ldots+d_{N-1})\le H.
\end{cases}
\end{align*}
Hence, the computation of $g_{\psi}$ requires solving an optimization of the form
\begin{align*}
\max \{\varepsilon\ge 0\,|\, E Y\le \frac{1}{\varepsilon}F(x),\,\,\forall Y \text{ with } A_b Y\le B_b\},
\end{align*}
with $Y:=\frac{1}{\varepsilon}[d_0^T,d_1^T,\ldots,d_{N-1}^T]^T$ being the free variables, where $A_b Y\le B_b$ represents the inequality $\|Y\|_\infty\le 1$ in matrix form. Hence, the result holds using Theorem~\ref{thm:farkas} in the Appendix.
 \end{proof}
  Let us remark that by defining $\psi_i =\nex^i \Gamma$, one gets $  \psi = \wedge_{i=0}^{N} \psi_i$.  Therefore, it follows from (iii) in Proposition~\ref{prop:struct} that  $g_{\psi}(x) = \min_i g_{\psi_i}(x)$. Hence, one can state the previous result in terms of reachability with exact time, with $g_{\psi_i}$ computed previously in Theorem~\ref{thm:lin_reachN}. 



\medskip
\subsection{Finite-Horizon Reachability}
In this section, we provide two approaches to compute the resilience metric for linear systems and finite-horizon reachability specifications. The first approach based on the exact time reachability approach proposed in Section~\ref{sec:linear_reach}, and a second approach based on scenario optimization.

\subsubsection{Computation using LMIs} Consider the linear system $\Sigma$ in (\ref{eq:linear}) with finite-horizon reachability specification $\psi = \lozenge^N \Gamma$ for some set $\Gamma$ as a polytope $\Gamma := \{x\in X\,|\, Gx\le H\}$. Then, we have  
\begin{align}
\label{eqn:opt_reach2}
g_\psi(x) &= \max\,\varepsilon\ge 0,\,s.t.\, \text{ for all }d_0,\ldots,d_{N-1}\in\ball_\varepsilon(0)\nonumber\\ 
&\begin{cases}
Gx\le H\text{ or } \\
G(Ax+d_0)\le H \text{ or } \\
G(A^2x+Ad_0+d_1)\le H \text{ or } \\
\vdots \\
G(A^Nx+A^{N-1}d_0+\ldots+d_{N-1})\le H.
\end{cases}
\end{align}
in view of (v) in Proposition~\ref{prop:struct}, one can select $\psi_i =\nex^i \Gamma$, to get $  \psi = \vee_{i=0}^{N} \psi_i$. Therefore, we can use the results of reachability in an exact time presented in Section~\ref{sec:linear_reach} to obtain a lower bound on the resilience metric given by $g_{\psi}(x) \ge \max_i g_{\psi_i}(x)$. However, using this approach, we do not have any measure of how the computed resilience metric is conservative. To deal with this issue, in the next section, we propose a scenario optimization approach to compute a lower bound on the resilience metric, together with a measure of its conservatism.

\medskip

\subsubsection{Computation using Scenario Optimization}
\label{sec:scenario}

To get a solution of the optimization problem in (\ref{eqn:opt_reach2}). We define $\mu=\frac{1}{\varepsilon}$, and the set $\mathcal{Z}=\{(z_0,z_1,\ldots,z_{N-1}) \in\ball_1(0)\}$. Hence, one gets that $g_{\psi}(x)$ is given by the following optimization problem:
\begin{align*}
g_\psi(x) = \min\,&\mu\ge 0,\,s.t.\, \text{ for all }z_0,\ldots,z_{N-1}\in\ball_1(0) \nonumber\\
&\begin{cases}
G\mu x\le \mu H\text{ or } \nonumber\\
G(\mu Ax+z_0)\le \mu H \text{ or } \nonumber\\
G(\mu A^2x+Az_0+z_1)\le \mu H \text{ or } \nonumber\\
\vdots \nonumber\\
G(\mu A^Nx+A^{N-1}z_0+\ldots+z_{N-1})\le \mu H.
\end{cases}
\end{align*}
If there exists a known $\overline{\mu} \geq 0$ such that the problem is feasible for $\overline{\mu}$\footnote{In this problem, $\bar{\mu}$ can be chosen as $\bar{\mu}=\frac{1}{\varepsilon}$ for any $\varepsilon>0$ that is solution to the optimization problem \eqref{eqn:opt_reach2}.}, the computation of $g_{\psi}(x)$ can be written in terms of the following non-convex program
\begin{align*}
    &\min \mu\\
    &\text{ s.t }~ \mu \in T \text{ and } \bigvee\limits_{k=1}^{N+1} g_k(\mu,z) \leq 0 \text{ for all } z \in \mathcal{Z},
\end{align*}
where $T=[0,\overline{\mu}]$ and for $\mu \geq 0$ and $z=(z_0,z_1,\ldots,z_{N-1}) \in \mathcal{Z}$, the map $g_k$, is defined by $g_1(\mu,z)= G\mu x- \mu H$ and for $k\in \{2,\ldots,N+1\}$
\begin{equation}
    g_k(\mu,z)= G(\mu A^{k-1}x+A^{k-2}z_0+\ldots+z_{k-2}) - \mu H
\end{equation}
The symbol $\bigvee\limits_{k=1}^{N+1} g_k(\mu,z) \leq 0$ indicates that there exists $k \in \{1,2,\ldots,N\}$ such that $g_k(\mu,z) \leq 0$, for all $z \in \mathcal{Z}$.

For a chosen parameter $\gamma \in \mathbb{R}^{N+1}_{\geq 0}$, in order to construct the $\gamma$-scenario program~\cite{esfahani2014performance} associated with the robust program defined above, we consider a uniform distribution on the space $\mathcal{Z}$ and obtain $M$ i.i.d. sample points  $\{z^1,z^2,\ldots,z^M\}$ with $z^i=(z^i_0,z^i_1,\ldots,z^i_{N-1}) \in \mathcal{Z}$, $i \in \{1,2,\ldots,M\}$. In view of~\cite{esfahani2014performance}, the $\gamma$-scenario program can be written as
\begin{align}
    &\min \mu \label{eqn:opt_nonlinSce2}\\
    &\text{ s.t }~ \mu \in T \text{ and }  \bigvee\limits_{k=1}^N  (g_k(\mu,z^i)+\gamma_k \leq 0),~~ i \in \{1,2,\ldots,M\}.\nonumber
\end{align}
We have the following result.
\begin{theorem}
\label{thm:SP}
Consider the linear system $\Sigma$ in (\ref{eq:linear}) and the specification $\psi=\lozenge^N \Gamma$, for some $N \in \mathbb{N}$. For $\beta \in (0,1]$ and $\eta \in [0,1]^{N+1}$, the optimal solution of \eqref{eqn:opt_nonlinSce2}, gives an approximate optimal solution to the optimization problem \eqref{eqn:opt_reach2} with confidence $(1-\beta)$, when the number of samples $M$ satisfies
\begin{equation}
\label{eqn:sample_size}
   M \geq M\left( (\eta_1^{\frac{1}{N}},\eta_2^{\frac{1}{N}},\ldots,\eta_{N+1}^{\frac{1}{N}}),\beta\right) 
\end{equation}
and $\gamma=(\gamma_1,\gamma_2,\ldots,\gamma_{N+1})$ with
$\gamma_k=(\|A\|^{k-1}\|Gx\|+\|H\|)\eta^{\frac{1}{N}}$, $k\in \{1,2,\ldots,N+1\}$, and the map $M:[0,1]^{N+1}\times (0,1] \rightarrow \mathbb{N}$ is defined for $\alpha \in [0,1]^{N+1}$ and $\theta \in (0,1]$ by $M(\alpha,\theta)=\min\{M\in \mathbb{N}\mid \sum_{k=1}^{N+1}(1-\alpha_k)^M \leq \theta \}$. Moreover, the optimal value $J^*$ for the optimization problem \eqref{eqn:opt_reach2}, and $J^*_M$ the optimal value for the optimization problem \eqref{eqn:opt_nonlinSce2} satisfy $J^*-J^*_N \leq \|\eta\|_{\infty}$, with confidence $(1-\beta)$.
\end{theorem}
\begin{proof}
We prove the result by using~\cite[Theorem 4.3]{esfahani2014performance}. We first provide a Lipschitz constant $L_{z,k}$ for the map $(\mu,z) \mapsto g_k(\mu,z)$ with respect to $\mu$, $k \in \{1,2,\ldots,N+1\}$. Consider $\mu,\mu' \in T$ and $z \in \mathcal{Z}$. We have:
\begin{align*}
    &\|g_k(\mu,z)-g_k(\mu',z)\| \\&\leq \|G(\mu A^{k-1}x+A^{k-2}z_0+\ldots+z_{k-2}) - \mu H \\&- G(\mu' A^{k-1}x+A^{k-2}z_0+\ldots+z_{k-2}) - \mu' H\| \\&+\|(\mu-\mu')H\| \\
    &\leq \|A\|^{k-1}\|(\mu-\mu')Gx\|+\|(\mu-\mu')H\| \\ &\leq (\|A\|^{k-1}\|Gx\|+\|H\|)|\mu-\mu'|.
\end{align*}
Hence, $L_{z,k}=\|A\|^{k-1}\|Gx\|+\|H\|$ is a Lipschitz constant for the map $g_k$, $k\in \{1,2,\ldots,N+1\}$ with respect to $z$. Moreover, for any $k \in \{1,2,\ldots,N+1\}$, the map$(\mu,z) \mapsto g_k(\mu,z)$ is linear and then convex with respect to $\mu$. Since the distribution on the set $\mathcal{Z}=\{(z_0,z_1,\ldots,z_{N-1}) \in\ball_1(0)\}$ is uniform, we choose $h(\varepsilon)=(\frac{\varepsilon}{2})^N$. The result follows from an application~\cite[Theorem 4.3]{esfahani2014performance},~\cite[Remark 3.9]{esfahani2014performance} and~\cite[Remark 3.5]{esfahani2014performance}.
\end{proof}


\section{Computation of Resilience for Nonlinear Systems}
\label{sec:nonlinear}

In this section, we extend the approaches to compute the resilience metric to the case of nonlinear systems. In particular, we discuss two approaches: (i) linear optimization-based approach and (ii) SMT-based approach.
\subsection{Linear Optimization-based Approach}
Since the computation of resilience for different types of specifications relies on reachability with the exact time as a building block, in this section, we focus on reachability with an exact time specification. The extension to other specifications can be done following the approaches presented in the previous section.

Consider the reachability at a specific time point: $\psi = \nex^N \Gamma$, for some set $\Gamma$ defined as a polytope $\Gamma = \{x\in X\,|\, Gx\le H\}$. We provide a linear optimization-based solution for computing $g_{\psi}(x)$ for nonlinear systems.
we use the following auxiliary result.

\begin{lemma}
\label{prop:lin_abstract}
Consider the nonlinear discrete-time system $\Sigma$ defined by:
\begin{equation}
\label{eqn:nonlinear}
x(k+1)=f(x(k)),~x(k) \in X \subseteq \mathbb{R}^n,~k\in \mathbb{N}.
\end{equation}
with a continuously differential map $f$. Assume without loss of generality that $f(0)=0$ and assume the existence of $\overline{\alpha}_{ij}$, $\underline{\alpha}_{ij}$, $i,j \in \{1,2,\ldots,n\}$, such that for all $x \in X\subset\reals^n$:
\begin{equation}
\label{eqn:jaco}
    \underline{\alpha}_{ij} \leq \frac{\partial f_i}{\partial x_j} \leq \overline{\alpha}_{ij},~ i,j \in \{1,2,\ldots,n\}
\end{equation}
Then, we have $\underline{A}x \leq f(x) \leq \overline{A}x$ for all $x \in X$,
with $\underline{A}, \overline{A} \in \mathbb{R}^{n\times n}$ defined for $i,j \in \{1,2,\ldots,n\}$ by $\underline{A}_{ij}=\underline{\alpha}_{ij}$ and $\overline{A}_{ij}=\overline{\alpha}_{ij}$.
\end{lemma}
\begin{proof}
Condition \eqref{eqn:jaco} implies that for all $x,y \in X$ and for all $i,j \in \{1,2,\ldots,n\}$, we have $$\sum\limits_{j=1}^n \underline{\alpha}_{ij}(x_j-y_j) \leq  f_i(x)-f_i(y) \leq \sum\limits_{j=1}^n\overline{\alpha}_{ij}(x_j-y_j).$$
In particular for $y=0$, one gets that for all $x \in X$ and for all $i\in \{1,2,\ldots,n\}$, we have $$\sum\limits_{j=1}^n \underline{\alpha}_{ij}x_j \leq  f_i(x) \leq \sum\limits_{j=1}^n\overline{\alpha}_{ij}x_j.$$
and the result holds.
\end{proof}

\medskip

\begin{theorem}
\label{thm:nonlin_lin}
Consider the nonlinear system $\Sigma$ given by
\begin{equation}
    x(k+1)=f(x(k))+d(k),~x(k) \in X \subseteq \mathbb{R}^n,~k\in \mathbb{N}
\end{equation}
and the specification $\psi = \nex^N \Gamma$, where $\Gamma$ is the polytope $\Gamma = \{x\in X\,|\, Gx\le H\}$. Assume the existence of $\overline{\alpha}_{i,j}$, $\underline{\alpha}_{ij}$, $i,j \in \{1,2,\ldots,n\}$, such that for all $x \in X$:
\begin{equation}
\label{eqn:jaco2}
    \underline{\alpha}_{ij} \leq \frac{\partial f_i}{\partial x_j} \leq \overline{\alpha}_{ij},~ i,j \in \{1,2,\ldots,n\},
\end{equation} 
Consider the matrices $B_1, B_{2}, \ldots, B_N \in \mathbb{R}^{n \times n}$ defined as
\begin{equation*}
    B_{k,ij}=\max\left(\sum_{h=1}^n G_{ih}\underline{D}_{k,hj},\sum_{h=1}^c G_{ih}\overline{D}_{k,hj}\right),
\end{equation*}
 where $B_{k,{ij}}$ represents the coefficient corresponding to the position $(i,j)$ of the matrix $B_k$, $k\in \{1,2,\ldots,N\}$, and the matrices $\underline{D}_k, \overline{D}_k$ are given with $\underline{D}_k= \underline{A}^k$ and $\overline{D}_k=\overline{A}^k$, with $\underline{A}, \overline{A}$ defined for $i,j \in \{1,2,\ldots,n\}$ by $\underline{A}_{ij}=\underline{\alpha}_{ij}$ and $\overline{A}_{ij}=\overline{\alpha}_{ij}$. Then, we have
\begin{align}
\label{eqn:linprog} \nonumber
   g_{\psi}(x) \geq \max &\{\varepsilon\ge 0\,|\, B_Nx+B_{N-1}d_0+\ldots+d_{N-1} \leq H \\  &  \text{ for all }d(0),\ldots,d(N-1)\in\ball_\varepsilon(0)\},
\end{align}
\end{theorem}
\begin{proof} 
The computation of resilience metric $g_{\psi}$ for reachability at a specific time point for the nonlinear system $\Sigma$ requires solving the following optimization problem
\begin{align}
\label{eqn:prob1}
&\max \{\varepsilon\ge 0\,|\, \nonumber \\ &Gf(f(\ldots((f(x)+d(0))+d(1)),\ldots)+d(N-1)) \leq H , \nonumber \\ &  \text{ for all }d(0),\ldots,d(N-1)\in\ball_\varepsilon(0)\}. 
\end{align}
In view of Lemma~\ref{prop:lin_abstract}, the nonlinear term $f(f(\ldots(f(x)+d(0))+d(1)),\ldots)+d(N-1))$ can be abstracted into a linear one as follows:
\begin{align*}
    &\underline{A}^Nx+\underline{A}^{N-1}d_0+\ldots+d_{N-1} \leq\\& f(f(\ldots((f(x)+d(0))+d(1)),\ldots)+d(N-1)) \leq\\& \overline{A}^Nx+\overline{A}^{N-1}d_0+\ldots+d_{N-1}.
\end{align*}
It follows that 
\begin{align*}
    Gf(f(\ldots&(f(x)+d(0))+d(1)),\ldots)+d(N-1)) \leq\\& B_Nx+B_{N-1}d_0+\ldots+d_{N-1},
\end{align*}
where $B \in \mathbb{R}^{n\times n}$ is defined for $i,j \in \{1,2,\ldots,n\}$ by $$B_{k,ij}=\max(\sum\limits_{h=1}^n G_{ih}\underline{A}^k_{hj},\sum\limits_{h=1}^c G_{ih}\overline{A}^k_{hj}),$$

where $B_{k,{ij}}$ represents the coefficient corresponding to the position $(i,j)$ of the matrix $B_k$, $k\in \{1,2,\ldots,N\}$. Hence, the optimization problem in \eqref{eqn:prob1} can be relaxed into the following linear optimization problem.
\begin{align}
\label{eqn:lin_abst}
    \max \{\varepsilon\ge 0\,|\,& B_Nx+B_{N-1}d_0+\ldots+d_{N-1} \leq H \nonumber \\&  \text{ for all }d(0),\ldots,d(N-1)\in\ball_\varepsilon(0)\},
\end{align}
which ends the proof.
\end{proof}

Theorem~\ref{thm:nonlin_lin} shows how to transform the problem of computing the resilience metric for nonlinear systems and exact-time reachability, into a problem for a linear system that can be resolved using the approach proposed in Section~\ref{sec:linear_reach}.

\subsection{SMT-based Approach}\label{Sec_SMT}
This section provides an SMT-based solution to compute resilience for nonlinear systems without any linear approximation. In addition, the approach also relaxes polytopic assumptions over subsets of state space corresponding to atomic propositions. To utilize an SMT-based solution, we consider the complement of the resilience definition \eqref{eq:quan} as given below. For a given LTL$_F$ formula $\psi$, the resilience in \eqref{eq:quan} can be equivalently written as
\begin{equation}\label{eq:resi}
\hspace{-0.2em} g_\psi(x) = \begin{cases}
\inf\left\{\varepsilon\ge 0\,|\,\xi(x,\varepsilon)\nvDash \psi\right\}, & \text{ if } \xi(x,0)\nvDash\psi\\
0, & \text{ if } \xi(x,0)\vDash\psi,
\end{cases}
\end{equation}
where $\neg\psi$ is a negation of the formula $\psi$.\\
Consider the set $\Gamma$ represented by logical conjunctions and disjunctions of $P$ nonlinear predicate functions $p_i:X\rightarrow\mathbb{R}$ as given below:
\begin{align}
    \Gamma=\{x\in X\mid p_i(x)\leq0 \ \Theta_i \  p_{i+1}(x)\leq0 \text{ is \textsf{true}},\nonumber\\ \forall i\in\{1,\ldots,P\}, \Theta_i\in\{\vee,\wedge\}\}.\label{new_set}
\end{align}
Now, consider LTL$_F$ formula $\psi=\nex^N\Gamma$, one can compute the resilience $g_\psi(x)$ in \eqref{eq:resi} by solving the following optimization problem:\\

$g_\psi(x) = \min\,\varepsilon\ge 0,\,s.t.\, \text{ there exists }d_0,\ldots,d_{N-1}\in\ball_{\varepsilon}(0) \text{ satisfying } \neg\big(p_i(\Psi)\leq 0 \Theta_i p_{i+1}(\Psi)\leq 0\big) \text{ is \textsf{true}}, \forall i\in\{1,\ldots,P\},\quad \quad \Theta_i\in\{\vee,\wedge\}$,
where $\Psi=f(f(\ldots((f(x)+d(0))+d(1)),\ldots)+d(N-1))$.\\

For a given value of $\varepsilon$, one can leverage SMT solvers such as dReal \cite{gao2013dreal} to check the feasibility of the condition. Then, the near minimum value of $\varepsilon$ can be found using the bisection method. In a similar way, one can formulate optimization problems for other formulas $\psi=\lozenge^N \Gamma$ and $\psi = \square^N \Gamma$.


\section{Case Studies}
\label{sec:case_studies}
\noindent\textbf{Temperature Regulation.}
We consider the problem of regulating the temperature in a circular building of $9$ rooms. The dynamics of the room temperatures are given by
\begin{align*}
{T}_i(k+1)& =T_i(k)+\alpha(T_{i+1}(k)+T_{i-1}(k)-2T_i(k))\\
&+\beta(T_e+\delta T_e-T_i(k)), \qquad i\in \{1,2,\ldots,9\},
\end{align*}
where $T_{i+1}$ and $T_{i-1}$ are the temperatures of the neighbor rooms (here $T_0=T_9$ and $T_{9+1}=T_1$), $T_e=0^\circ C$ is the outside temperature, considered as a disturbance and $\alpha$ and $\beta$ are the conduction factors.
The numerical parameters are taken from~\cite{saoud2018contract} and given by $\alpha=0.45$ and $\beta=0.045$.

We consider two scenarios for this example. In the first scenario, the desired behavior of the system is as follows: The temperatures of the $9$ rooms initiated in the set $X_{0}=[24,25]^9$ should reach the target set $X_T=[21,22]^9$ exactly at $N=3$ steps while remaining in the safe set $X_S=[20.5,25]^9$. This behavior can be described by the LTL$_F$ formula:
\begin{equation}
    \psi=\psi_1 \wedge \psi_2, \text{ with } \psi_1=\square^3X_S \text{ and } \psi_2=\nex^3X_T.
\end{equation}
The objective is to compute the range of admissible external disturbances $\delta T_e$ under which any trajectory of the system initiated in the set $X_{0}$ satisfies $\psi$. Since the system is linear and the set of initial states $X_0$ is convex, we rely on Theorems~\ref{thm:lin_reachN},~\ref{thm:fin_safe} and~\ref{thm:conv} to compute the resilience metric. The numerical implementations show that the resilience metric is given by the set $\Delta_{T_e}=[-2.23, 2.23] ^\circ C$.

In the second scenario, the desired behavior of the system is as follows: the temperatures of the $9$ rooms initiated in the set $X_{0}=[50,51]^9$ should remain in the safe set $X_S=[10,51]^9$, reach the target set $X_{T_1}=[41,43]^9$ exactly at $N=4$ steps. Moreover, once the set $X_{T_2}=[31,35]^9$ is reached, the temprature should reach the target set $X_{T_3}=[14,28]^9$ exactly at $N=11$ steps. This behavior can be described by the LTL$_F$ formula: $ \psi=\psi_1 \wedge \psi_2 \wedge \psi_3$, with
\begin{equation}
\psi_1=\square^2X_S, \psi_2=\nex^4X_{T_1} \text{ and } \psi_3=X_{T_2} \implies \nex^{11} X_{T_3}.
\end{equation}
The objective is to compute the range of admissible external disturbances $\delta T_e$ under which any trajectory of the system initiated in the set $X_{0}$ satisfies $\psi$. The numerical implementations show that the resilience metric is given by the set $\Delta_{T_e}=[-3.12, 3.12] ^\circ C$.

Figure~\ref{fig:temp_reg1} (top) shows the nominal trajectories (with $\delta T_e=0$) of the system for the $9$ rooms. Figure~\ref{fig:temp_reg1} (bottom) shows the trajectories of the system for the $9$ rooms with a disturbance $\delta T_e$ randomly chosen in the set $\Delta_{T_e}=[-3.12, 3.12] ^\circ C$. 
\begin{figure}[!t]
	\begin{center}
		\includegraphics[scale=0.49]{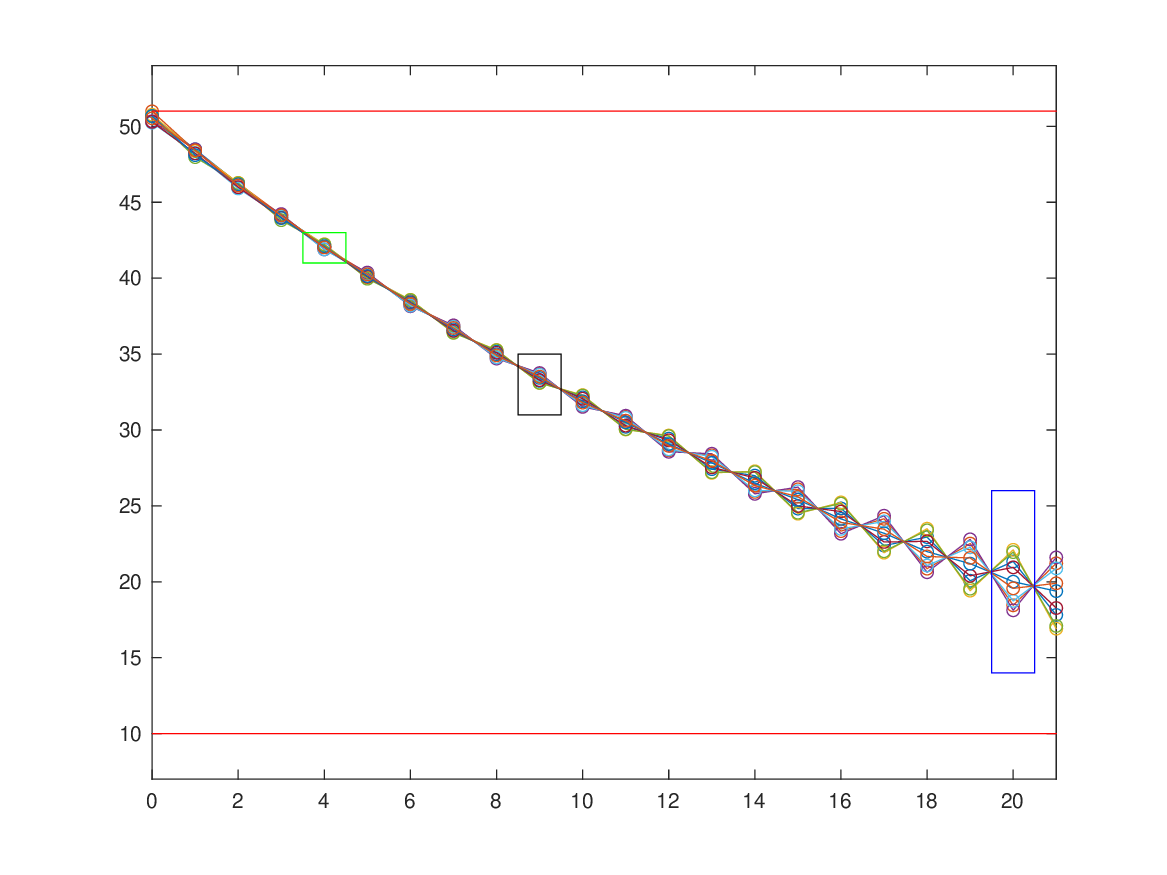}
        \includegraphics[scale=0.49]{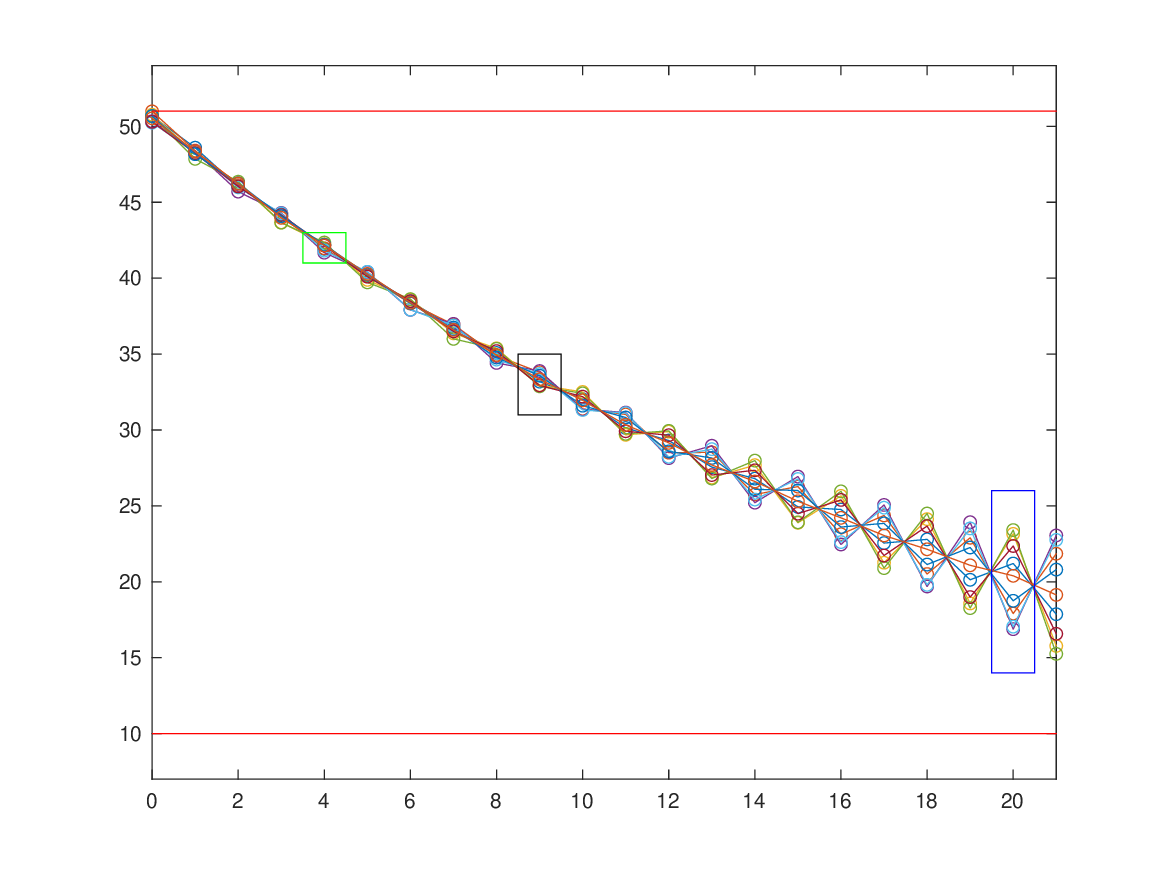}
	\end{center}
	\vspace{-2em}
	\caption{Evolution of the temperatures in the nine rooms with a disturbance $\delta T_e=0$ (top) and with a disturbance $\delta T_e$ randomly chosen in the set $\Delta_{T_e}=[-3.12, 3.12] ^\circ C$ (bottom). The red boundaries represent the safe set $X_S$. The green region represents the set $X_{T_1}$. The black region represents the set $X_{T_2}$ and the blue region represents the set $X_{T_3}$.}
	\label{fig:temp_reg1}
 \vspace{-1em}
\end{figure}


\smallskip
\noindent\textbf{Adaptive Cruise Control.}
Consider a vehicle moving along a straight road. The dynamics of the vehicle is adapted from~\cite{saoud2018contract} and given by the following difference equation:
\begin{equation}
\label{eqn:model}
v(k+1)=v(k)+\frac{\tau}{m}(f_{0}-f_1v(k)-f_{2}v(k)^2),
\end{equation}
where $m>0$ is the mass of the vehicle, $v \geq 0$ represents the velocity of the vehicle and the term $f_0-f_1-f_2v^2$ includes the rolling resistance and aerodynamics and $\tau$ represents a sampling period. Moreover, we include a lead vehicle whose velocity $v_0\geq 0$ is constant. The dynamics of the system are
\begin{equation*}
\label{eqn4}
\left\{
\begin{array}{r c l}
h(k+1) &=& h(k)+\tau(v_0+\delta v_0 - v(k))\\
v(k+1) &=& v(k)+\frac{\tau}{m}(f_{0}+\delta f_0-f_1v(k)-f_{2}v(k)^2),
\end{array}
\right.
\end{equation*}
where $\delta v_0$ is the uncertainty on the velocity $v_0$ of the lead vehicle and $\delta f_0$ is the uncertainty on the parameter $f_0$. The desired behavior can be described by the following LTL$_F$ formula:
    $\psi:=\psi_1 \wedge \psi_2$
    with $\psi_1:=\square^{4} X_S$ and $\psi_2:= \nex^4 X_{T_1}$. This behavior can be interpreted as follows: the relative position should remain in the safe set $X_S=[h_{S,\min},h_{S,\max}]\times[v_{S,\min},v_{S,\max}]$ and the relative position and velocity should reach the set $X_{T_1}=[h_{T_1,\min},h_{T_1,\max}]\times[v_{T_1,\min},v_{T_1,\max}]$ in $4$ steps.

The objective is to compute the resilience metric under which the trajectory of the system initiated from $x_{0}=(60,15)$ satisfies $\psi$. The numerical values of the vehicle parameters and the considered specifications are given in Table \ref{table:parameters}.

We use the approaches developed in Section~\ref{sec:nonlinear} to deal with this nonlinear system. First, we used linear optimization based approach. One can easily check that the values of the parameters $\alpha_{i,j}$, $i,j \in \{1,2\}$, given by $\underline{\alpha}_{11}=\overline{\alpha}_{11}=1,$ $\underline{\alpha}_{12}=\overline{\alpha}_{12}=-1$, $\underline{\alpha}_{21}=\overline{\alpha}_{21}=0$, $\underline{\alpha}_{22}=1-\frac{\tau}{m}-2f_2v_{S,\max}$ and $\overline{\alpha}_{22}=1-\frac{\tau}{m}-2f_2v_{S,\min}$ satisfy the inequalities in \eqref{eqn:jaco2}, where the bounds on $\underline{\alpha}_{22}$ and $\overline{\alpha}_{22}$ follows from the fact that we are interested in dealing with velocities of the following vehicle within the interval $[v_{S,\min},v_{S,\max}]$. Hence, in view of Theorem~\ref{thm:nonlin_lin}, one can use the linear program in \eqref{eqn:linprog} to compute an approximation of the resilience metric $g_{\psi}$. The numerical implementations show that the value of the resilience metric is given by $g_{\psi}=0.0042$. Next, we computed resilience $g_\psi$ using SMT solver dReal and a bisection approach as discussed in Section \ref{Sec_SMT}. The obtained value of $g_\psi$ using SMT based approach is $0.0121$ which is less conservative compared to that of using linearization based approach. The corresponding values for the admissible resilience metric on the variable $\delta v_0$ and $\delta f_0$ are given by, $\Delta v_0=[-0.012,0.012]$ and $\Delta f_0=[-16.6,16.6]$.


\subsection{DC Motor}
Our third case study is a $2$-dimensional discrete-time DC motor adapted from~\cite{adewuyi2013dc} as follows:
\begin{align*}
x_1(k+1) &= x_1(k) + \tau\big (\frac{-R}{L}x_1(k) - \frac{k_{dc}}{L} x_2(k)\big ),\\
x_2(k+1) &= x_2(k) + \tau\big (\frac{-k_{dc}}{J}x_1(k) - \frac{b}{J} x_2(k)+d(k)\big ),
\end{align*}
where $x_1, x_2, R = 1, L = 0.5$, and $J = 0.01$ are the armature current, the rotational speed of the shaft, the electric resistance, the electric inductance, and the moment of inertia of the rotor, respectively. In addition, $\tau = 0.01, b = 0.1,$ and $K_{dc} = 0.01$ represent both the motor torque and the back electromotive force.
The desired behavior of the system is as follows: the armature current $(x_1)$ and the rotational speed of the shaft $(x_2)$ starting from initial state $x_0=(0.4,0.5)$ should reach the target set $X_T=[0.1,0.35] \times [0.1,0.35]$ within $3$ steps. The corresponding LTL$_F$ formula is $\psi=\lozenge^3 X_T$. The objective is to compute the maximal value of the disturbance $d$ under which any trajectory of the system satisfies the specification $\psi$. 

We use the Scenario approach proposed in Section~\ref{sec:scenario}, with $\eta=(0.01,0.01,0.01,0.01)$ and $\beta=0.01$. To achieve the chosen confidence and accuracy, we choose a sampling size $M$ according to Theorem~\ref{thm:SP}. The numerical implementations show that, with confidence $(1-\beta)$, the resilience metric is given by $g_{\psi}(x_0)=0.0485$. 

Figure~\ref{fig:motor} shows examples of $1000$ trajectories under disturbance trajectories chosen randomly on the interval $[-g_{\psi}(x_0),g_{\psi}(x_0)]^2$. In green, we represent the boundaries of the target set $X_T$. One can readily see that almost all the trajectories satisfy the specification $\psi$. Since the obtained value of the resilience metric $g_{\psi}(x_0)$ is valid with confidence $(1-\beta)$, one can also see that some trajectories fail to satisfy the specification $\psi$.

\begin{figure}[!t]
\centering
\includegraphics[scale=0.48]{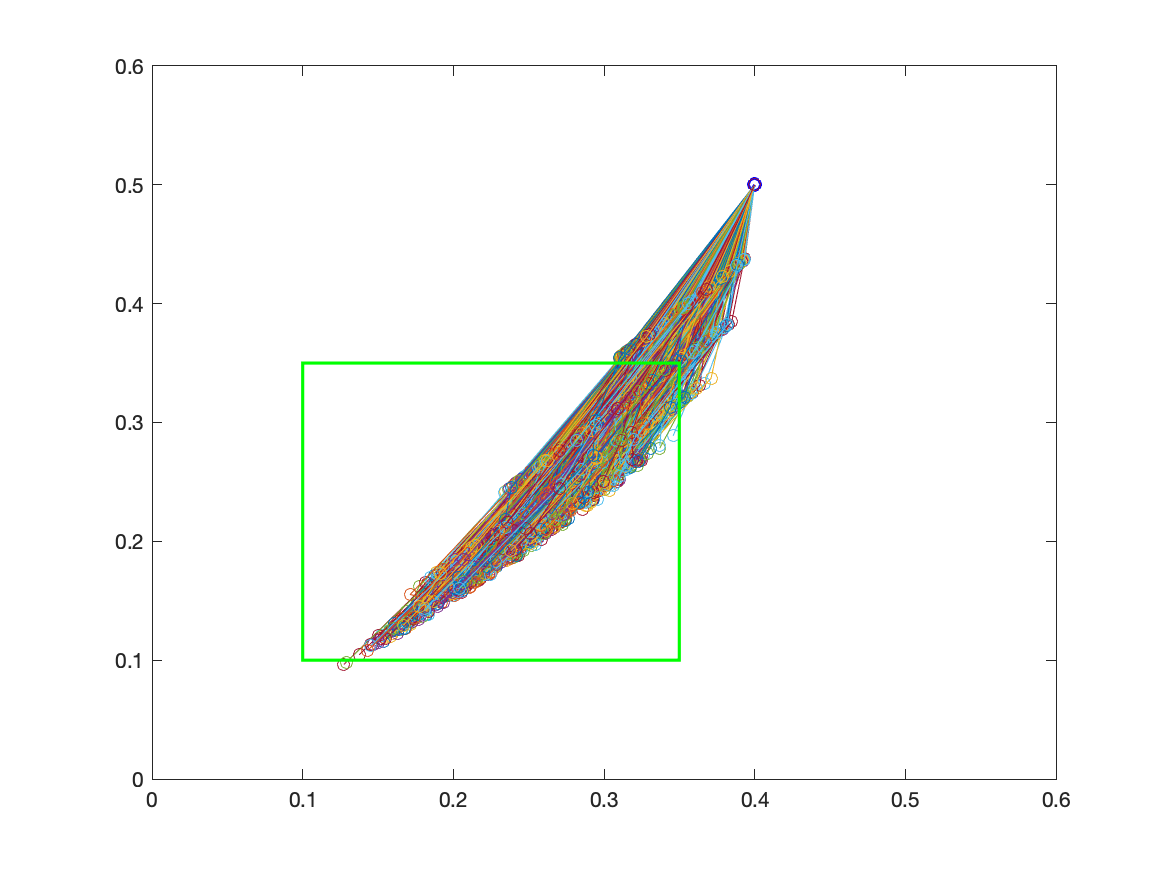}
        \vspace{-0.5cm}
	\caption{Evolution of $1000$ perturbed trajectory starting from the initial condition $x_0=(0.4,0.5)$. The green boundaries represent the target set $X_T$.}
 \label{fig:motor}	
 \vspace{-3ex}
\end{figure}

\section{Conclusions and Discussion}
\label{sec:concl}
We provided a new resilience metric for cyber-physical systems that integrates the dynamical evolution of the system with temporal logic requirements. We showed how this resilience metric can be computed for linear and nonlinear discrete-time models of the system and instances of linear temporal logic specifications. The related computations utilize Farkas' Lemma, linear programs, scenario optimization and SMT-based approaches. A more detailed analysis is provided for the new introduced classes of closed and convex finite linear temporal logic specifications. Illustrating our findings, we showcase computational demonstrations spanning temperature regulation in buildings, adaptive cruise control, and DC motors.

In the future, we plan to develop techniques for enhancing resilience (i.e., maximizing resilience over decision variables) and extending the ideas to continuous-time dynamical systems. We also plan to use the proposed resilience metric for the design of contracts for large-scale networked systems \cite{saoud2018composition} by specifying the largest class of disturbances that a subsystem in the network can tolerate.

Our new notion of resilience can be used by other research communities working on the safety and security of critical infrastructures and safety-critical systems. For instance, it will be useful for designing minimally-invasive load-altering attacks \cite{katewa2021optimal} and providing a vulnerability map with respect to logical requirements on the behavior of the power network \cite{wooding2020formal}. 
%

\begin{table}
	\caption{{Vehicle and safety parameters}}
	\centering
	\begin{tabular}{|c|c|c|}
		\hline 
		Parameter & Value & Unit \\ 
		\hline 
		$M$ & $1370$ & $Kg$ \\ 
		
		$f_0$ & $51.0709$ & $N$ \\ 
		
		$f_1$ & $5$ & $Ns/m$ \\
		
		$f_2$ & $0.4161$ &  $Ns^2/{m^2}$\\

		

		
		
		$h_{T_1,\min}$ & $60.2$ &  $m$\\
		
		$h_{T_1,\max}$ & $60.5$ &  $m$\\
		
		$v_{T_1,\min}$ & $14.6$ &  $m/s$\\
		
		$v_{T_1,\max}$ & $14.9$ &  $m/s$\\

		
		
		
		
		$h_{S,\min}$ & $59.7$ &  $m$\\
		
		$h_{S,\max}$ & $60.7$ &  $m$\\
		
		$v_{S,\min}$ & $14.5$ &  $m/s$\\
		
		$v_{S,\max}$ & $15.5$ &  $m/s$\\

		\hline 
	\end{tabular}
	\label{table:parameters}
\end{table}



\section*{APPENDIX}

\subsection{Auxilliary results}
The following theorem is borrowed from the book \cite[Corollary 7.lh]{schrijver1998theory} and is known as the affine form of Farkas' lemma. 

\begin{theorem}[\cite{schrijver1998theory}]
Let the system $Ax\leq b$ of linear inequalities have at least one solution, and suppose that the linear inequality $cx\leq\delta$ holds for each $x$ satisfying $Ax\leq b$. Then for some $\delta'\leq \delta$ the linear inequality $cx\leq \delta'$ is a nonnegative linear combination of the inequalities in the system $Ax\leq b$.
\end{theorem}

\begin{theorem}
\label{thm:farkas}
Suppose the set $Ex\leq F$ is not empty. The following two statements are equivalent:
\begin{itemize}
    \item $Ex\leq F$ holds for all $x$ with $Ax\le b$;
    \item There exists a non-negative matrix $P$ such that $P A = E$ and $Pb\le F$.
\end{itemize}
\end{theorem}

\begin{proof}
We use the above theorem for each row of $Ex\leq F$ indicated by $E_ix\leq F_i$. Then there is nonnegative row vectors $\lambda_i\ge 0$ such that $\lambda_i A = E_i$ and $\lambda_i b \le F_i$.
The non-negative constants in $\lambda_i$ are referred to as the Farkas multipliers. Putting all these multipliers as row vectors in a matrix $P$ we get that $P A = E$ with $Pb\le F$
\end{proof}

\begin{lemma}
\label{lem:continuity}
Consider the discrete-time system $\Sigma$ in \eqref{eqn:sys}. If the map $f:\mathbb{R}^n \rightarrow \mathbb{R}^n$ is continuous, then for any $x \in \mathbb{R}^n$, and for any sequence $\varepsilon_i \geq 0$, $i \in \mathbb{N}$, the map $\xi(x,.):\mathbb{R}_{\geq 0} \rightrightarrows \mathbb{R}^n$, defined in \eqref{eqn:reach} satisfies: 
\begin{equation}
\label{eqn:cond2}
\lim_{i \rightarrow \infty }\xi(x,\varepsilon_i) \subseteq \xi(x,\lim_{i \rightarrow \infty } \varepsilon_i)
\end{equation}
\end{lemma}
\begin{proof}
The proof follows easily from the continuity of the map $f$ and the fact that $\lim_{i \rightarrow \infty }\ball_{\varepsilon_i}(0) = \ball_{\lim_{i \rightarrow \infty }\varepsilon_i}(0)$
\end{proof}

The condition in \eqref{eqn:cond2} can be seen as a sufficient condition for the outer semicontinuity~\cite{aubin2009viability} of the map $\xi$ with respect to its second variable.

\subsection{Closed specifications}
\label{sec:closed}

In this part, we provide some characterizations of the concept of closed specifications introduced in Definition \ref{def:closed}. We first relate the closedness of the atomic proposition $p$ to the closedness of the corresponding state space $L^{-1}(p) \subseteq X$.

\begin{proposition}
Consider the alphabet $\Sigma_a = 2^{AP}$ with the labeling function $L:X \rightarrow \Sigma_a$.
Then the atomic proposition $p\in AP$ is closed if the set $L^{-1}(p):=\cup\{L^{-1}(A) | A\in \Sigma_a \text{ such that } p\in A\}$ is a closed subset of $X$. Moreover, the specification $\neg p$, with $p$ being the atomic proposition, is closed if the set $L^{-1}(p)$ is an open subset of $X$.
\end{proposition}
\begin{proof}
    Consider a converging sequence of trajectories $w_{x,i}=(x_{0,i},x_{1,i},\ldots,x_{N-1,i})$, $i \in \mathbb{N}$, such that for all $i \in \mathbb{N}$, $w_{x,i} \vDash \psi$. Hence, one gets that for all $i \in \mathbb{N}$, $x_{0,i} \in L^{-1}(\psi)$. It follows from the closedness of the set $L^{-1}(\psi)$, that $x_0=\lim_{i\rightarrow \infty}x_{0,i} \in L^{-1}(\psi)$ which in turn implies that $w_x=\lim_{i\rightarrow \infty} w_{x,i} \vDash \psi$ and that $\psi$ is a closed specification. The proof of the second result follows immediately.
\end{proof}

Now, we provide a characterization of the fragement of LTL$_f$ specifications that are closed.

\begin{proposition}
Let $\psi_1, \psi_2$ be closed specifications and $M \in \mathbb{N}$, with $M \leq N$, where $N$ represents the length of the trajectories of the LTL$_F$ specifications $\psi_1, \psi_2$. Let $\psi$ be constructed with the grammar
\begin{align*}
    \psi :=   &\psi_1 \,|\,\psi_1 \wedge \psi_2\,|\,\psi_1 \vee \psi_2 \,|\,\nex \psi_1 \,|\,\nex^M \psi_1 \,|\,\square \psi_1  \,|\,  \square^M \psi_1  \,|\,\lozenge \psi_1 \,|\,\lozenge^M \psi_1.
    \end{align*}
Then, the specification $\psi$ is a closed specification.
\end{proposition}
\begin{proof}
We provide proof for each property separately.

\textbf{\underline{Proof for $\psi_1 \wedge \psi_2$}:} Consider a converging sequence of trajectories $w_{x,i}$, $i \in \mathbb{N}$, such that for all $i \in \mathbb{N}$, $w_{x,i} \vDash \psi_1$ and $w_{x,i} \vDash \psi_2$, then from closedness of the specifications $\psi_1$ and $\psi_2$, it follows that $w_x=\lim_{i\rightarrow \infty} w_{x,i} \vDash \psi_1$ and $w_x \vDash \psi_2$, which implies that $w_x\vDash \psi_1 \wedge \psi_2$.

\textbf{\underline{Proof for $\psi_1 \vee \psi_2$}:} Consider a converging sequence of trajectories $w_{x,i}$, $i \in \mathbb{N}$, such that for all $i \in \mathbb{N}$, $w_{x,i} \vDash \psi_1$ or $w_{x,i} \vDash \psi_2$. Thus, we can use this fact to construct two subsequences: the subsequence $w_{x,j_k}$ where $w_{x,j_k}$ is a term of $w_{x,i}$ satisfying $w_{x,j_k} \vDash \psi_1$, and the subsequence $w_{x,h_k}$ where $w_{x,h_k}$ is a term of $w_{x,i}$ satisfying $w_{x,h_k} \vDash \psi_2$. First, since all the terms of the subsequence $w_{x,j_k}$ satisfy $w_{x,j_k} \vDash \psi_1$, one gets from the closedeness of $\psi_1$ that $w^1=\lim_{k\rightarrow \infty}w_{x,j_k} \vDash \psi_1$. Similarly, one gets that $w^2=\lim_{k\rightarrow \infty}w_{x,h_k} \vDash \psi_2$. Now, since the sequence $w_{x,i}$, $i \in \mathbb{N}$, is converging and has one limit, the limit of $w_{x,i}$ is either $w^1$ or $w^2$. Hence, one gets that $w_x=\lim_{i\rightarrow \infty} \vDash \psi_1 \vee \psi_2$.

   \textbf{\underline{Proof for $\nex^M \psi_1$}:} Consider a converging sequence of trajectories $w_{x,i}$, $i \in \mathbb{N}$, such that for all $i \in \mathbb{N}$, $w_{x,i} \vDash \nex^M \psi_1$, with $w_{x,i}=(x_{0,i},x_{1,i},\ldots,x_{N-1,i})$. Hence, one gets that for all $i \in \mathbb{N}$, $x_{M,i} \vDash \psi_1$, using the fact that $w_x=\lim_{i\rightarrow \infty} w_{x,i}=x_{0},x_{1},x_{2},\ldots,x_N$, and that $x_M=\lim_{i\rightarrow \infty} x_{M,i}$, once gets from the closedness of the specification $\psi_1$ that $x_M=\lim_{i\rightarrow \infty} x_{M,i} \vDash \psi_1$, which in turn implies that $w_{x} \vDash \nex^M \psi_1$. The proof for the case of $\nex \psi_1$ can be derived similarly.

\textbf{\underline{Proof for $\square^M \psi_1$}:} Consider a converging sequence of trajectories $w_{x,i}$, $i \in \mathbb{N}$, such that for all $i \in \mathbb{N}$, $w_{x,i} \vDash \square^M \psi_1$, with $w_{x,i}=(x_{0,i},x_{1,i},\ldots,x_{N-1,i})$. Hence, one gets that for all $i \in \mathbb{N}$ and for all $j \in \{0,1,\ldots M-1\}$, $x_{j,i} \vDash \psi_1$. Using the fact that $w_x=\lim_{i\rightarrow \infty} w_{x,i}=x_{0},x_{1},x_{2},\ldots,x_N$, and that $x_j=\lim_{i\rightarrow \infty} x_{j,i}$, $j \in \{0,1,\ldots,M-1\}$, once gets from the closedness of the specification $\psi_1$ that $x_j=\lim_{i\rightarrow \infty} x_{j,i} \vDash \psi_1$, $j \in \{0,1,\ldots,M-1\}$, which in turn implies that $w_{x} \vDash \square^M \psi_1$. The proof for the case of $\square \psi_1$ can be derived similarly.

\textbf{\underline{Proof for $\lozenge^M \psi_1$}:} Consider a converging sequence of trajectories $w_{x,i}$, $i \in \mathbb{N}$, such that for all $i \in \mathbb{N}$, $w_{x,i} \vDash \lozenge^M \psi_1$. Hence, for all $i \in \mathbb{N}$, we have the existence of $j \in \{0,1,\ldots,M-1\}$ such that $w_{x,i} \vDash \nex^j \psi_1$. Thus, we can use this fact to construct $M$ subsequence, where the kth subsequence $w_{x,k_h}$ is made of the terms of $w_{x,i}$ satisfying $w_{x,k_h} \vDash \nex^k\psi_1$, $k \in \{0,1,\ldots,M-1\}$. First, since all the terms of the subsequence $w_{x,k_h}$ satisfy $w_{x,k_h} \vDash \nex^k\psi_1$, one gets from the closedenss of the specification $\psi_1$ and the fact that the closedeness property is preserved under $\nex^k$ that $w^k=\lim_{h\rightarrow \infty}w_{x,k_h} \vDash \nex^k\psi_1$. Now, since the sequence $w_{x,i}$, $i \in \mathbb{N}$, is converging and has one limit, we have the existence of $k \in \{0,1,\ldots,M-1\}$ such that $w_x=\lim_{i\rightarrow \infty}=w^k$. Using the fact that $w^k \vDash \nex^k\psi_1$ for all $k \in \{0,1,\ldots,M-1\}$ one gets that $w_x\vDash \lozenge^M \psi_1$.
\end{proof}

\subsection{Convex specifications}
\label{sec:convex}

In this part, we provide some characterizations of the concept of convex specifications introduced in Definition \ref{def:convex}. We first relate the convexity of the atomic proposition $p$ to the convexity of the corresponding state space $L^{-1}(p) \subseteq X$.

\begin{proposition}
Consider the alphabet $\Sigma_a = 2^{AP}$ with the labeling function $L:X \rightarrow \Sigma_a$.
Then the atomic proposition $p\in AP$ is convex if the set $L^{-1}(p):=\cup\{L^{-1}(A) | A\in \Sigma_a \text{ such that } p\in A\}$ is a convex set.
\end{proposition}
\begin{proof}
    Consider the trajectories $w_{x,i}$, $i \in \{1,2\}$, such that $w_{x,i} \vDash \nex^M \psi $ and consider $\lambda \in [0,1]$. Define $w_x=\lambda w_{x,1} + (1-\lambda) w_{x,2}$. For $w_{x,i}=(x_{0,i},x_{1,i},\ldots,x_{N-1,i})$, $i=1,2$, one gets that $x_{0,i} \vDash \psi$. It follows from the convexity of the set $L^{-1}(\psi)$, that $x_0=\lambda x_{0,1} + (1-\lambda) x_{0,2}$ which in turn implies that $w_x= \lambda w_{x,1} + (1-\lambda) w_{x,2}\vDash \psi$ and that $\psi$ is a convex specification. 
\end{proof}

Now, we provide a characterization of the fragment of LTL$_f$ specifications that are convex.

\begin{proposition}
Let $\psi_1, \psi_2$ be convex specification and $M \in \mathbb{N}$, with $M \leq N$, where $N$ represents the length of the trajectories of the LTL$_F$ specifications $\psi_1, \psi_2$. Let $\psi$ be constructed with the grammar
\begin{eqnarray*}
    \psi :=  \psi_1 \,|\,\psi_1 \wedge \psi_2\,|\,\nex \psi_1 \,|\,\nex^M \psi_1 \,|\,\square \psi_1 \,|\,\square^M \psi_1. 
    \end{eqnarray*}
Then, the specification $\psi$ is a convex specification.
\end{proposition}
\begin{proof}
\textbf{\underline{Proof for $\psi_1 \wedge \psi_2$}:} Consider the trajectories $w_{x,i}$, $i \in \{1,2\}$, such that $w_{x,i} \vDash \psi_1 \wedge \psi_2$ and consider $\lambda \in [0,1]$. Define $w_x=\lambda w_{x,1} + (1-\lambda) w_{x,2}$. Since both $\psi_1$ and $\psi_2$ are convex, we have that $w_{x} \vDash \psi_1$ and $w_{x} \vDash \psi_2$ and the result holds.

\textbf{\underline{Proof for $\nex^M \psi_1$}:} Consider the trajectories $w_{x,i}$, $i \in \{1,2\}$, such that $w_{x,i} \vDash \nex^M \psi_1 $ and consider $\lambda \in [0,1]$. Define $w_x=\lambda w_{x,1} + (1-\lambda) w_{x,2}$. For $w_{x,i}=(x_{0,i},x_{1,i},\ldots,x_{N-1,i})$, $i=1,2$, one gets that $x_{M,i} \vDash \psi_1$. Since $\psi_1$ is a convex specification, it follows that $w_{x,N}=\lambda x_{M,1} + (1-\lambda) x_{M,2}$, with $w_x=x_1,x_2,\ldots,x_N$, which in turn implies that $w_{x} \vDash \nex^M \psi_1$. The proof for the case of $\nex \psi_1$ can be derived similarly.

\textbf{\underline{Proof for $\square^M \psi_1$}:} Consider the trajectories $w_{x,i}$, $i \in \{1,2\}$, with $w_{x,i}=(x_{0,i},x_{1,i},\ldots,x_{N-1,i})$ and such that $w_{x,i} \vDash \square^M \psi_1 $ and consider $\lambda \in [0,1]$. Define $w_x=\lambda w_{x,1} + (1-\lambda) w_{x,2}$. Hence, one gets that for $i =1,2$ and for all $j \in \{0,1,\ldots M-1\}$, $x_{j,i} \vDash \psi_1$. Using the fact that $\psi_1$ is a convex specification, it follows that the trajectory $w_x=(x_{0},x_{1},\ldots,x_{N-1})$ satisfies for $j \in \{0,1,\ldots M-1\}$ that $x_j=\lambda x_{j,1}+(1-\lambda) x_{j,2} \vDash \psi_1$, which in turn implies that $w_{x} \vDash \square^M \psi_1$. The proof for the case of $\square \psi_1$ can be derived similarly.
\end{proof}

\bibliographystyle{alpha}
\bibliography{biblio}

\end{document}